\documentclass[10pt,journal,compsoc]{IEEEtran}

\usepackage{caption}
\usepackage[none]{hyphenat}
\usepackage{algpseudocode}
\usepackage{amsmath}
\usepackage{amsfonts}
\usepackage{amssymb}
\usepackage{subfigure}
\usepackage{graphicx}
\usepackage[table]{xcolor}
\usepackage{array,multirow}

\usepackage{multicol}

\usepackage{amsfonts}
\usepackage{booktabs}
\usepackage{siunitx}

\usepackage{color}

\usepackage{tikz} 
\newcommand*\circled[1]{\tikz[baseline=(char.base)]{%        
        \node[shape=rectangle,fill=gray!20,draw,inner sep=2pt,opacity=0.5,text opacity=1] (char) {#1};}}

\newcommand\citem{%
  \stepcounter{enumi}\item[\circled{\arabic{enumi}}]}

\usepackage{enumitem}

\usepackage{url}
\makeatletter
\def\url@leostyle{%
  \@ifundefined{selectfont}{\def\UrlFont{\sf}}{\def\UrlFont{\small\ttfamily}}}
\makeatother
\newcommand{\eat}[1]{}
\usepackage[linesnumbered,ruled]{algorithm2e}
\usepackage{mdframed} 

\renewcommand{\IEEEbibitemsep}{1pt plus 2pt}
\makeatletter
\IEEEtriggercmd{\reset@font\normalfont\small}
\makeatother
\IEEEtriggeratref{1}

\usepackage{xcolor}
\definecolor{light-gray}{gray}{0.9}

\newenvironment{packed_enum}{%
  \begin{enumerate}%
  }{\end{enumerate}}

\usepackage{amsthm}

\newtheorem{lemma}{Lemma}

\newcolumntype{L}[1]{>{\raggedright\let\newline\\\arraybackslash\hspace{0pt}}m{#1}}

\begin{document}

\title{SilentDelivery: Practical Timed-delivery of Private Information using Smart Contracts}
% \title{}

\author{Chao~Li,~\IEEEmembership{Member,~IEEE,}
        and~Balaji~Palanisamy,~\IEEEmembership{Member,~IEEE}% <-this % stops a space
\IEEEcompsocitemizethanks{\IEEEcompsocthanksitem C. Li is with the Beijing Key Laboratory of Security and Privacy in Intelligent Transportation, Beijing Jiaotong University, Beijing,
China.\protect\\
% note need leading \protect in front of \\ to get a newline within \thanks as
% \\ is fragile and will error, could use \hfil\break instead.
E-mail: li.chao@bjtu.edu.cn
\IEEEcompsocthanksitem B. Palanisamy is with the School of Computing and Information, University of Pittsburgh, Pittsburgh,
USA.}% <-this % stops an unwanted space
% \thanks{Manuscript received April 19, 2005; revised August 26, 2015.}
}

\IEEEtitleabstractindextext{%
\begin{abstract}

% ================= short ======================
This paper proposes \textit{SilentDelivery}, a secure, scalable and cost-efficient protocol for implementing timed information delivery service in a decentralized blockchain network. \textit{SilentDelivery} employs a novel combination of threshold secret sharing and decentralized smart contracts. 
The protocol maintains shares of the decryption key of the private information of an information sender using a group of mailmen recruited in a blockchain network before the specified future time-frame and restores the information to the information recipient at the required time-frame.
To tackle the key challenges that limit the security and scalability of the protocol,
\textit{SilentDelivery} incorporates two novel countermeasure strategies.
The first strategy, namely \textit{silent recruitment}, enables a mailman to get recruited by a sender \textit{silently} without the knowledge of any third party.
The second strategy, namely \textit{dual-mode execution}, makes the protocol run in a lightweight mode by default, where the cost of running smart contracts is significantly reduced.
We rigorously analyze the security of \textit{SilentDelivery} and implement the protocol over the Ethereum official test network. 
The results demonstrate that \textit{SilentDelivery} is more secure and scalable compared to the state of the art and reduces the cost of running smart contracts by 85\%. 

\end{abstract}

}

\maketitle

\IEEEdisplaynontitleabstractindextext

\IEEEpeerreviewmaketitle

\IEEEraisesectionheading{\section{Introduction}\label{sec:introduction}}

% \section{Introduction}

\IEEEPARstart{R}{apid} advancements in Internet and service technology has led to a proliferation of information exchange happening in the cyberspace.
%With the rapid acceptance of going paperless, fueled by the ongoing advances in data transmission technology, 
%we are witnessing an increasing amount of information sharing and exchange happening daily in the cyberspace.
\underline{T}imed \underline{i}nformation \underline{d}elivery \underline{s}ervice (\textit{TIDS}) refers to a class of service that enables an information sender to make a piece of private information arrive at an information recipient during a chosen future time-frame.
% without disclosing the information before {\bf or after} the time-frame. {\bf Do we need to include "after" the timeframe also?}
% \textcolor{red}{we may change 'without...time-frame' to 'rather than before or after the time-frame'}
%As a common way of exchanging information, \underline{t}imed \underline{i}nformation \underline{d}elivery \underline{s}ervice (\textit{TIDS}) refers to a type of service that an information sender could rely on to make a piece of private information arrive at an information recipient within a future time-frame specified by the sender, without disclosing the information either before or after the time-frame.
Many scenarios require timed delivery of information in real-world. 
For example, courier services allow clients to select a predetermined time-frame during which their mail can be delivered.\footnotemark\ 
The ability of \textit{TIDS} to control the arrival time of sensitive information and knowing precisely when information arrives can be crucial for many businesses and enterprises.\footnotemark\ 
Imagine a situation in which Alice would like her business proposal to arrive at the corporate headquarter exactly during the board meeting time. Here, an early arrival of the proposal could potentially leak her idea to her competitors while a late arrival will remove Alice's proposal out of the competition.
%When information is transmitted in electronic form, it seems as if the nearly instantaneous network transmission make the senders easily handle the arrival time of information.
% the senders can easily handle the arrival time of information due to the almost instantaneous network transmission.
% the arrival time of information could be easily handled by the senders themselves due to the almost instantaneous network transmission. 
%However, there might be many factors that stop a sender from sending out her information on time, such as a time-zone difference, a network issue or simply a bad memory.
While there are numerous services (e.g., Boomerang~\cite{boomerang} and Postfity~\cite{postfity}) that provide pre-scheduled timed delivery of information, current implementations of timed information delivery services (\textit{TIDS}) are heavily centralized. These services require the users to entirely trust the centralized servers and their security properties are solely limited to a single point of trust.
% not maliciously colluding with potential adversaries and secretly disclosing their information.
More importantly, even in scenarios when the service providers are trustworthy, the services are still prone to unpredictable security breaches or insider attacks that are beyond the control of the service providers~\cite{chen2018certchain,hu2018searching}.
On the other hand, the emergence of Blockchain technologies such as Ethereum~\cite{buterin2014next} and Smart contracts~\cite{wood2014ethereum} provides significant potential for new security designs that support a decentralized implementation of \textit{TIDS} to overcome the single point of trust issues associated with centralized approaches.

% Relying on witness encryption~\cite{garg2013witness}, one distinct decentralized approach~\cite{liu2018build,liu2015time,liu2015extractable} encloses private information with the blockchain puzzles used in Proof-of-Work~\cite{nakamoto2008bitcoin}, which relieves information recipients from the burden of computing puzzles because the blockchain puzzles are periodically solved by blockchain miners.
% However, the current implementation of witness encryption is far from practical, making this approach suffer from high performance overhead~\cite{liu2015time,ning2018keeping}. 
% Another problem is that using similar blockchain puzzles for protecting many pieces of information could potentially lead to

\footnotetext[1]{\begin{scriptsize}E.g., UPS customers can pick a 2-hour time frame for ensuring a confirmed delivery. https://www.ups.com/us/en/help-center/sri/ups-my-choice-delivery\\-window.page \end{scriptsize}}

\footnotetext[2]{\begin{scriptsize} https://www.rapidparcel.com/timed-delivery/ \end{scriptsize}}

%\footnotetext[3]{\begin{scriptsize} Boomerang (https://www.boomeranggmail.com/) allows users of Gmail to schedule their emails to be sent when users have no connection with the Internet. \end{scriptsize}}

%\footnotetext[4]{\begin{scriptsize} Postfity (https://postfity.com/) helps users schedule messages to be posted onto social networks at a predetermined time frame.  \end{scriptsize}}
  
In this paper, we present \textit{SilentDelivery}, a secure, scalable and cost-efficient protocol for implementing timed information delivery services (\textit{TIDS}) in a decentralized blockchain network.  
\textit{SilentDelivery} employs a novel combination of threshold secret sharing~\cite{shamir1979share} and decentralized smart contracts~\cite{wood2014ethereum}. 
The protocol maintains shares of the decryption key of the private information of an information sender using a group of mailmen recruited in a blockchain network before the specified future time-frame and restores the information to the information recipient at the required time-frame.
Here, the benefits of employing blockchains and smart contracts are threefold:
(a) \textit{SilentDelivery} establishes a permissionless \textit{TIDS} marketplace atop well-established public blockchains (e.g., Ethereum), allowing any interested party to join the service community as a mailman. 
Information senders and mailmen can thus simply leverage the native cryptocurrencies of blockchains for payment and settlement;
(b) \textit{SilentDelivery} inherits the two key security properties \textit{persistence} and \textit{liveness} of public blockchain backbone~\cite{garay2015bitcoin} and hence, both senders and mailmen cannot deny any information (e.g., recruitment relationship) recorded in the blockchain (\textit{persistence}). They are protected from selective denial of service attacks and can always change the blockchain state with valid transactions (\textit{liveness}); 
(c) Inspired by recent work on blockchain-based secure multi-party computations~\cite{kumaresan2014use,kumaresan2016amortizing,kumaresan2016improvements}, \textit{SilentDelivery} protects the fairness of \textit{TIDS} and handles any misbehavior against the protocol by using the cryptocurrencies pre-locked in smart contracts as security \textit{deposits} and their confiscation as \textit{penalty}. Hence, rational participants of the protocol are always incentivized to honestly follow the protocol.
However, a naive design of the protocol, where each message is delivered by a single mailman, may result in both low service availability (i.e., the single mailman accidentally loses data or network connection) as well as lead to a single point of trust (i.e., anyone that manages the decryption key may collude with the single mailman). 
Therefore, \textit{SilentDelivery} delivers each message with a group of $n$ mailmen. Each mailman manages a share of the decryption key as well as a copy of the encrypted message. As long as at least $t\ (t<n)$ mailmen manage shares well, the message would be successfully delivered without getting prematurely disclosed.

% However, a naive design of the protocol, where each message is delivered by a single mailman, encounters a security challenge that is essentially a key management problem.
% Specifically, to protect the private information carried by the message against the single mailman, the sender needs to encrypt the message with a secret key, which should be later revealed to the information recipient exactly during the prescribed time-frame.

% Blockchain and smart contracts
% (1) cryptocurrency, service market
% (2) persistance liveness, foundamental security
% (3) MPC wih penalty

% Secret share
% (1) private information encryption
% (2) 

% \textcolor{red}{
% We identify two key challenges that significantly limit the security and scalability of the protocol:
% (1) for the sake of fair trade, the protocol demands senders and mailmen to conclude recruitment relationships via smart contracts but the transparency of smart contracts makes it difficult to conceal the relationships before the future time-frame, which endangers the service security in several aspects and 
% (2) due to the use of threshold secret sharing, a higher service availability requires a sender to recruit more mailmen, resulting in higher cost for carrying out the interactions between mailmen and smart contracts.
Here, the use of blockchains and smart contracts also leads to two key challenges that impact the security and scalability of the protocol. 
First, for the sake of fair trade, the protocol requires senders and mailmen to conclude recruitment relationships via smart contracts but the transparency of smart contracts makes it difficult to conceal the relationships before the future time-frame which challenges the service security in multiple aspects. 
Second, due to the use of threshold secret sharing, the protocol requires senders to recruit a group of $n$ mailmen to gain higher service availability which involves $O(n)$ transaction fees for carrying out the interactions between the $n$ recruited mailmen and smart contracts. 
% \textcolor{red}{
Hence, given a large $n$, the cost for running the service could be more than what most users can afford, making the approach difficult to work in practice.
% }
% \textit{SilentDelivery} tackles these challenges through two novel countermeasure strategies.
% To tackle the key challenges that limit the security and scalability of the protocol,
\textit{SilentDelivery} incorporates two novel countermeasure strategies to tackle these two challenges, respectively.
The first strategy, namely \textit{silent recruitment}, recruit each mailman for a sender \textit{silently} without the knowledge of any third party while still making it possible for the recruitment relationship to be revealed to the smart contracts during a future time-frame.
% \textcolor{red}{
Moreover, a mailman recruited via \textit{silent recruitment} cannot prove the hidden recruitment relationship to any party, including the information recipient and other mailmen. Hence, \textit{silent recruitment} could significantly reduce the attack surface by making any potential adversary unable to differentiate a small group of mailmen that deliver a specific message from other mailmen inside the large mailmen community who have registered in the smart contracts.
% }
% helps expand the attack surface encountered by any potential adversary from a selected public group of mailmen to the entire mailmen community registered in the smart contracts.
The second strategy, namely \textit{dual-mode execution}, makes the protocol run in a lightweight mode by default, where the non-scalable regulations are cut off to reduce the cost of running smart contracts significantly.
When a dispute occurs, any recruited mailman reserves the ability to switch the protocol to a heavyweight mode by rebinding the removed regulations with smart contracts to redress and penalize any fraudulent or dishonest behavior, just as if these regulations were never decoupled.

% \textcolor{red}{
In a nutshell, this paper makes the following key contributions:
\begin{itemize}[leftmargin=*]
\item To the best of our knowledge, \textit{SilentDelivery} is the \textit{first} practical decentralized approach designed for \textit{TIDS} that is secure, scalable and cost-efficient. 
\item \textit{SilentDelivery} is a \textit{pure} decentralized approach designed for a trustless environment, without requiring any trusted party.
\item \textit{SilentDelivery} completely \textit{isolates} the service execution from the state of sender after the service has been set up, without requiring any assistance from the sender side.
\item \textit{SilentDelivery} employs the \textit{silent recruitment} strategy to reduce the attack surface of \textit{TIDS} and employs the \textit{dual-mode execution} strategy to reduce the service cost.
% incentivizes mailmen to stay honest so that the non-scalable operations in the protocol could be executed without interacting with the blockchain, which reduces the cost of recruiting a group of $n$ mailmen from $O(n)$ to $O(1)$.
\item We rigorously analyze the security of \textit{SilentDelivery} and implement the protocol over the Ethereum official test network. 
The results demonstrate that \textit{SilentDelivery} is more secure and scalable compared to the state of the art and reduces the cost of running smart contracts by 85\%.
\end{itemize}
% }

% \textcolor{red}{
The rest of this paper is organized as follows: 
We start by introducing preliminaries in Section~2.
In Section~3, we present a strawman protocol for \textit{TIDS} that involves public recruitment relationships as well as a high service cost. We then analyze the two challenges and present the attack methods and design goals.
In Section~4, we propose the \textit{SilentDelivery} protocol and presents the two countermeasure strategies, namaly \textit{silent recruitment} and \textit{dual-mode execution}.
In Section~5, we rigorously analyze the security and cost of \textit{SilentDelivery}.
We implement and evaluate the protocols over the Ethereum official test network in Section~6.
Finally, we discuss related work in Section~7 and conclude in Section~8.

\section{Preliminaries}

In this section, we discuss the preliminaries about smart contracts and introduce the key cryptographic tools used in our work.
While we discuss smart contracts in the context of Ethereum~\cite{wood2014ethereum}, we note that our solutions are also applicable to a wide range of other smart contract platforms. 
We summarize the key notations that will be used in this section and in the rest of this paper in Table~\ref{t1}.

% \subsection{Account types}

\subsection{Smart contracts}
There are two types of accounts in Ethereum, namely External Owned Accounts (EOAs) and Contract Accounts (CAs).
To interact with the Ethereum blockchain, an user needs to own an EOA by locally creating a pair of keys.
Specifically, the public key $pk_{EOA}$ can generate a 20-byte address $addr(EOA)$ to uniquely identify the EOA and the private key $sk_{EOA}$ can be used by the user to sign transactions or other types of data.
Then, any user can create a smart contract by sending out a contract creation transaction from a controlled EOA.
The 20-byte address $addr(CA)$ of the created smart contract is generated in a deterministic and predictable way and becomes the unique identity of the contract account (CA).

A smart contract (or contract, $C$) in Ethereum is a piece of program created using a high-level contract-oriented programming language such as \textit{Solidity}\footnote{\begin{scriptsize} https://github.com/ethereum/solidity \end{scriptsize}}.
After compiling into a low-level bytecode language called Ethereum Virtual Machine (EVM) code, the created contract is packaged into a transaction, which is then broadcasted to the entire Ethereum network formed by tens of thousands of miner nodes.
Following the Proof-of-Work (PoW) consensus protocol~\cite{nakamoto2008bitcoin}, all the miners in Ethereum competitively solve a blockchain puzzle and the winner packages the received transactions into a block and appends the new block to the end of Ethereum blockchain.
From then on, it is hard to tamper with the contract as each miner maintains a copy of the new block and an adversary has to falsify majority of these copies in order to change the network consensus about the contract. 
% A smart contract, after being deployed into Ethereum blockchain via the above process, becomes ready to be interacted.
In Ethereum, any user can create a transaction to call any accessible function within a deployed contract.
The function called by this transaction is then executed and verified by the miners and its inputs and outputs are both recorded in Ethereum blockchain.
In other words, smart contracts in Ethereum are executed transparently in a decentralized manner and the results are deterministic.

\begin{table}
\begin{center}
\begin{tabular}{|c|p{6.5cm}|}
\hline
\textbf{notation} & \textbf{description} \\
\hline
$S$ & an information sender \\
$M$ & a mailman \\
$R$ & an information recipient \\
$C$ & a smart contract \\
$C.fun()$ & function $fun()$ within contract $C$ \\
% $\rightrightarrows$ & broadcast messages via public off-chain channels\\
$\dashrightarrow$ & transmit messages via private off-chain channels\\
$\Rightarrow$ & function invocation that results in transaction fees\\
$addr(*)$ & an 20-byte address of an EOA or a CA\\
$hash(*)$ & an keccak hash value\\
$List(*)$ & a list of objects\\
\hline
\end{tabular}
\end{center}
\caption{Summary of notations.}      
% \vspace{-5mm}
\label{t1}
\end{table}

\subsection{Transaction fees}
\label{s3.3}
In order to either deploy a new contract or call a deployed contract in Ethereum, one needs to spend Gas, or transaction fees.
Based on the computational work of the transactions or smart contracts executed by miners, a part of Ether\footnote{\begin{scriptsize} The native cryptocurrency in Ethereum, denoted by $\mathsf{\Xi}$. \end{scriptsize}} needs to be spent in order to purchase an amount of Gas, which is then paid to the miner that creates the new block. 
The Gas system is important for Ethereum as it helps incentivize miners to stay honest, nullify denial-of-service attacks and encourage efficiency in smart contract programming.
On the other hand, the Gas system requires protocols, especially the multi-party ones, to be designed with higher scalability in Ethereum. This is due to the fact that even a single-round multi-party protocol could spend a lot of money to run in case of a large number of participants. 
% makes a request for more scalable and efficient protocol design that uses smart contracts, especially for multi-party protocols, because even single-round multi-party protocols could spend a lot of money to run in case of a large number of participants.  

% \textcolor{red}{
% \subsection{Transaction fees}
% \label{s3.3}
% In order to either deploy a new contract or call a deployed contract in Ethereum, one needs to spend Gas, or transaction fees.
% Based on the complexity of the contract or that of the called function, an amount of ether needs to be spent to purchase an amount of Gas as a transaction fee, which is then paid to the winning miner. 
% The Gas system is important for Ethereum as it helps to incentivize miners to stay honest, to nullify denial-of-service attacks and to encourage efficiency in smart contract programming.
% On the other hand, the Gas system requires protocols, especially the multi-party ones, to be designed with higher scalability in Ethereum. This is due to the fact that even a single-round multi-party protocol could spend a lot of money to run in case of a large number of participants. 
% }

% \textcolor{red}{
\subsection{Off-chain channels}
\label{s3.4}
In Ethereum, nodes forming the underlying P2P network can send messages to each other via off-chain channels established through the Whisper protocol~\cite{Whisper2017}.
By default, messages are publicly broadcast to the entire P2P network.
The Whisper protocol enables two important functionalities:
(1) message filtering functionality that enables an EOA to set up a filter to only accept interested messages marked with a specific 4-byte topic (e.g., the topic could be \textit{TIDS});
(2) private channel functionality that enables two EOAs to establish private off-chain channels to exchange messages with either symmetric or asymmetric keys.
Concretely, any EOA can locally generate a pair of asymmetric Whisper keys and reveal the public Whisper key by storing it into the transparent \textit{TIDS} smart contract, which allows other EOAs to initialize private off-chain channels.
In the rest of this paper, we assume that the private off-chain channels are secure and we omit the detailed settings of off-chain channels. 
Also, to make it easier to distinguish on/off-chain interactions, we denote private off-chain communication between two EOAs and public on-chain function invocation transactions with symbols $\dashrightarrow$ and $\Rightarrow$, respectively.
It is worth noting that off-chain communication costs no money while on-chain transactions charge transaction fees.
% }

\subsection{Cryptographic tools}

The design of \textit{SilentDelivery} employs several key cryptographic tools:
(1) \textit{(t,n)-threshold secret sharing}~\cite{shamir1979share} is used to split the decryption key of the private information into $n$ shares, among which any $M$ shares could recover the key but $t-1$ or fewer shares fail to do that. Specifically, we denote the key split and key restoration as $shares \gets SS.split(key,[t,n])$ and $key \gets SS.restore(shares,[t,n])$, respectively.
(2) we use the Keccak 256-bit hash function supported by Ethereum and it is denoted as $hash(*)$.
(3) we use the ECDSA signature supported by Ethereum. 
Specifically, a EOA (i.e., \textit{signer}) can sign any message with its private account key $sk_{EOA}$ via \textit{JavaScript} API and get signature $vrs \gets sig(hash(message))$.
Later, other EOAs or CAs can recover the address of the \textit{signer} EOA (i.e. $addr(signer)$) via \textit{JavaScript} API or \textit{Solidity} native function and get
$addr(signer) \gets vf(hash(message),vrs)$. 
% \textcolor{red}{
The signature $vrs$ can be either privately delivered via off-chain channels or publicly announced via the blockchain.  
% }

\section{A Strawman Protocol}
% A straw-man proposal is a brainstormed simple draft proposal intended to generate discussion of its disadvantages and to provoke the generation of new and better proposals.
We first describe the timed information delivery service (\textit{TIDS}) as a three-phase process. We then propose a strawman protocol that implements \textit{TIDS} using threshold secret sharing and smart contracts. We analyze the key limitations of the strawman protocol in terms of its security and scalability and finally, we present the attack methods.
 % and enumerate the design goals of the \textit{SilentDelivery} protocol, which is presented in Section~\ref{s5}.

\subsection{TIDS as a three-phase process}
We describe the \textit{TIDS} problem as a three-phase process:

\begin{itemize}[leftmargin=*]

\item \textit{TIDS.send}: 
The information sender (or sender, $S$) sends her private information  with a time-frame and the identity of information recipient (or recipient, $R$) to the \textit{TIDS} service provider (i.e., mailmen in our protocol).
\item \textit{TIDS.pend}:
The \textit{TIDS} provider preserves the private information before the specified time-frame.
\item \textit{TIDS.deliver}:
The \textit{TIDS} provider delivers the private information to the recipient during the specified time-frame.

% \item \textit{TIDS.send}: 
% The information sender (or sender) encrypts her private information with a key through symmetric encryption, pass the encrypted information to the information recipient (or recipient), recruits a group of mailmen and assign each mailman a share of the key.
% \item \textit{TIDS.pend}:
% Each mailman maintains the assigned share before the future time-frame specified by sender.

% \item \textit{TIDS.deliver}:
% Each mailman reveal the maintained share to recipient, who then restore the key and get access to the private information.

\end{itemize}

\subsection{The strawman protocol}

We first propose a strawman protocol that implements \textit{TIDS} using threshold secret sharing and smart contracts.
We sketch the strawman protocol in Fig.~\ref{protocol_sketch_01} and present the formal description in Fig.~\ref{protocol_detail_1}.
Specifically, the regulations of the strawman protocol are programmed as an agent contract $C_{agent}$, through which a sender can recruit a group of Ethereum accounts to jointly take the role of the \textit{TIDS} provider.
We name each account serving for \textit{TIDS} a mailman, $M$.
The protocol demands each mailman $M$ to register itself to the agent contract $C_{agent}$ via $newMailman(whiskey,\mathsf{\Xi} deposit)$, where 
% $List(pubkey)$ denotes a list of public keys serving different future time-frames, 
$whiskey$ is a public Whisper key used to establish a private communication channel with the mailman using the Ethereum Whisper protocol and 
$\mathsf{\Xi} deposit$ represents the amount of Ether($\mathsf{\Xi}$) that will be locked in $C_{agent}$ as a security deposit.
We can understand that there is a recruitment agreement at $C_{agent}$, which goes into effect only when it has been signed by both a sender and a mailman, so that the registration of a mailman in this context means that the mailman has signed on a recruitment agreement and has promised to serve for future senders without violating the protocol. Otherwise the $\mathsf{\Xi} deposit$ will get confiscated.
We next describe the three phases of the strawman protocol in detail.

\begin{figure}
\centering
{
   
    \includegraphics[width=9cm,height=5.5cm]{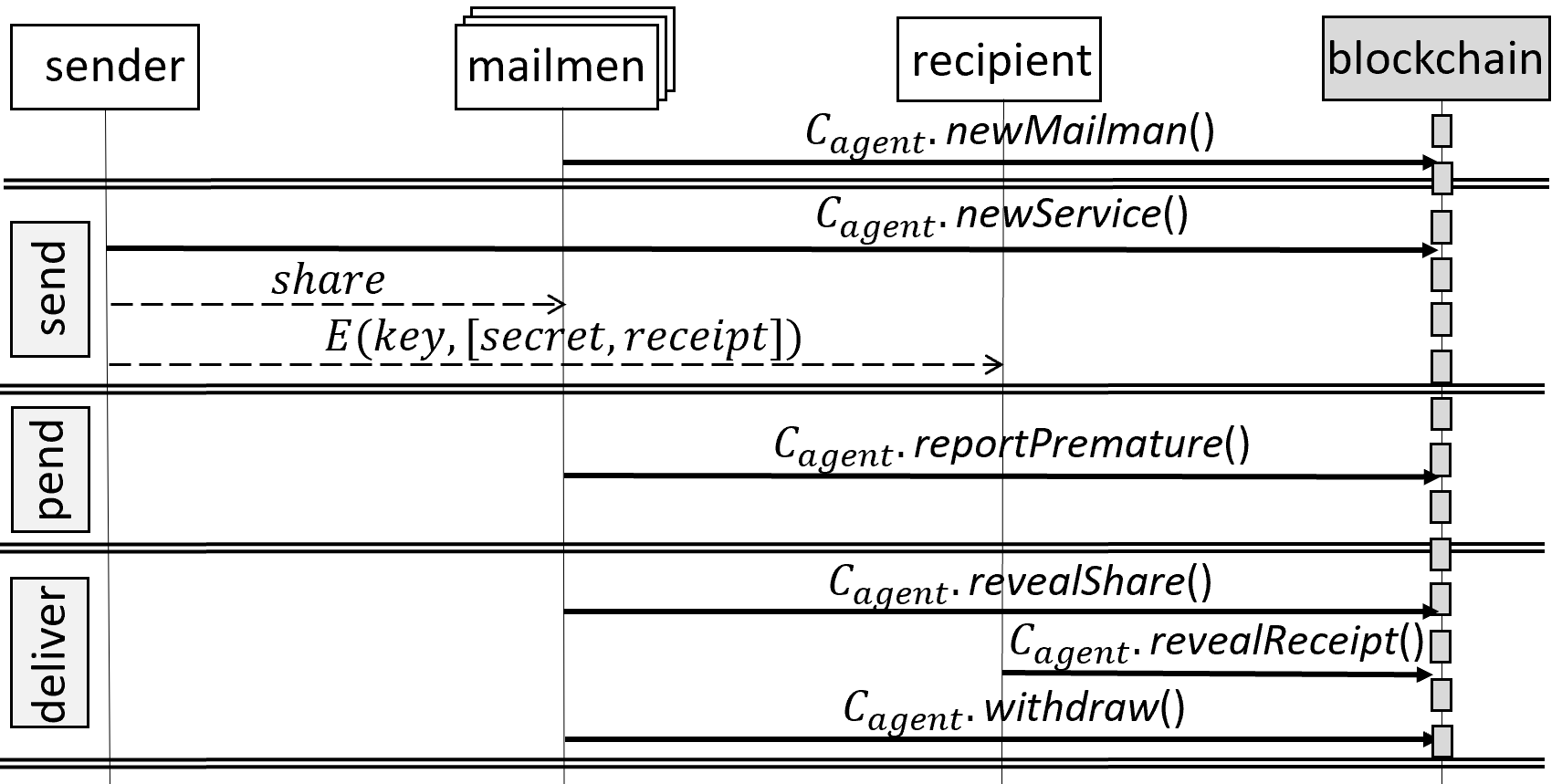}
}
%\vspace{-6mm}
\caption {Strawman protocol sketch. Solid lines denote calling on-chain smart contract functions. Dashed lines denote private off-chain communication.}
%\vspace{-5mm}
\label{protocol_sketch_01} 
\end{figure}
 
\noindent \textbf{\textit{TIDS.send}}: 
Sender $S$ first creates a $key$ as well as a $receipt$, splits $key$ into $shares$ via secret sharing and encrypts both her private information (or $info$) and $receipt$ with $key$.
Sender $S$ then sets up a new service with $C_{agent}$ via $newService()$ and specifies the service details including a future time-frame $[day,slot]$, parameters $[t,n]$ for secret sharing and addresses (i.e., $addr(*)$) of both recipient $R$ and a list (i.e., $List(*)$) of randomly selected registered mailmen.
In addition, sender $S$ also commits the values of $shares$ and $receipt$ by submitting their hash values (i.e., $hash(*)$) to $C_{agent}$ and locks an amount of Ether as $\mathsf{\Xi} remuneration$ to pay the mailmen.
Finally, through private off-chain channels, sender $S$ assigns each recruited mailman a $share$ and transmits the encrypted $[info,receipt]$ to recipient $R$. 
Here, the submission of addresses of selected mailmen implies that sender $S$ has signed the recruitment agreements with these mailmen and has promised to pay $\mathsf{\Xi} remuneration$ for a successful service.
From then on, 
% since both sides have signed these agreements, 
the agreements signed by both sides come into force. 

\begin{figure}
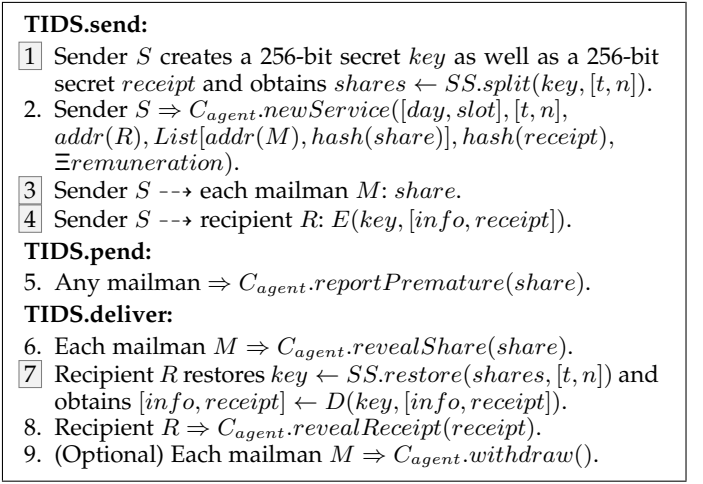

\begin{minipage}{0.5\textwidth}
\begin{small}
\begin{mdframed}[innerleftmargin=8pt]

% \centerline{\textbf{Silent recruitment}}
% \noindent \textbf{Before TIDS.send:} 
% \begin{packed_enum}[leftmargin=*]
%   \item Each mailman $M \Rightarrow C_{agent}.newMailman(key_{whisper},\\\$ deposit)$.
% \end{packed_enum}
\noindent \textbf{TIDS.send:} 
\begin{packed_enum}[leftmargin=*]
  % \item Each mailman $M \Rightarrow C_{agent}.newMailman(key_{w})$.
  \citem Sender $S$ creates a 256-bit secret $key$ as well as a 256-bit secret $receipt$ and obtains $shares \gets SS.split(key,[t,n])$.
  \item Sender $S \Rightarrow C_{agent}$.$newService([day,slot],[t,n],\\
  addr(R),List[addr(M),hash(share)],hash(receipt), \\
  \mathsf{\Xi} remuneration)$.
  \citem Sender $S \dashrightarrow$ each mailman $M$: $share$.
  \citem Sender $S \dashrightarrow$ recipient $R$: $E(key,[info, receipt])$. 
  % \begin{packed_enum}[leftmargin=*]
  %   \item $S \dashrightarrow$ each mailman $M$: $share$.
  %   \item $S \dashrightarrow$ recipient $R$: $E(key,[secret, receipt])$.
  % \end{packed_enum}
\end{packed_enum}
\noindent \textbf{TIDS.pend:} 
\begin{packed_enum}[leftmargin=*]
  \setcounter{enumi}{4}
  \item Any mailman
  $\Rightarrow C_{agent}.reportPremature(share)$.
\end{packed_enum}
\noindent \textbf{TIDS.deliver:} 
\begin{packed_enum}[leftmargin=*]
  \setcounter{enumi}{5}
  \item Each mailman $M \Rightarrow C_{agent}.revealShare(share)$.
  \citem Recipient $R$ restores $key \gets SS.restore(shares,[t,n])$ and obtains $[info,receipt] \gets D(key,[info, receipt])$.
  \item Recipient $R \Rightarrow C_{agent}.revealReceipt(receipt)$.
  \item (Optional) Each mailman $M \Rightarrow C_{agent}.withdraw()$.
\end{packed_enum}
% \noindent \textbf{After TIDS.deliver:} 
% \begin{packed_enum}[leftmargin=*]
%   \item Each mailman $M \Rightarrow C_{agent}.withdraw()$.
% \end{packed_enum}

\end{mdframed}
\end{small}
\end{minipage}

\captionof{figure}{
% \begin{footnotesize} 
The strawman protocol. 
The symbol $\Rightarrow$ denotes calling a function of a smart contract. The symbol $\dashrightarrow$ denotes transmitting data through a private channel. 
A step with a gray bullet (e.g., $\colorbox{gray!18}{1}$) refers to an off-chain action not recorded by blockchain while a step with a white bullet (e.g., 2) refers to an on-chain action recorded by blockchain.
% \end{footnotesize}
}
%\vspace{-4mm}
\label{protocol_detail_1}
\end{figure}

\noindent \textbf{\textit{TIDS.pend}}: 
If any recruited mailman discloses its maintained $share$ before the time-frame, any mailman obtaining this $share$ will be able to report the premature disclosure of this $share$ to $C_{agent}$ via $reportPremature()$.~\footnote{
\begin{scriptsize} 
In case that a party getting the disclosed $share$ owns no mailman account, the party can report the disclosure by informing the $share$ to a mailman and share the reward with the mailman. The fairness can be guaranteed by deploying a smart contract that takes both disclosed $share$ from the party and an amount of deposit from the mailman as inputs. The smart contract can verify the input $share$ with $hash(share)$ retrieved from $C_{agent}$ and then both reward the party with the locked deposit and send out a transaction to call $reportPremature()$ that rewards (larger than deposit) the mailmen.
\end{scriptsize}}
With $hash(share)$ stored in $C_{agent}$, this function will be able to verify this report and divide $\mathsf{\Xi} deposit$ of the accused mailman to the informer as well as sender $S$ in case of a $true$ report. 
% \textcolor{red}{
Here, the reward is split for the purpose of preventing malicious mailman from intentionally reporting itself with a different account and get back its locked deposit before disclosing the maintained $share$ to other parties. 
As a result, a malicious mailman can only obtain a negative payoff of performing premature disclosure unless it colludes with sender $S$ or receives additional payment (i.e., bribery). We believe that the collusion share does not happen as it brings no benefits for sender $S$. We will discuss the latter situation as the bribery attack in section~3.4.
% }
% \footnote{
% \begin{scriptsize} 

% \end{scriptsize}}

\noindent \textbf{\textit{TIDS.deliver}}: 
During the time-frame, each mailman reveals its maintained $share$ to $C_{agent}$ via $revealShare()$.
Recipient $R$, after restoring $key$ and decrypts $[info,receipt]$, receives $info$ and notifies the arrival of $info$ by revealing $receipt$ to $C_{agent}$ via $revealReceipt()$.
% \textcolor{red}{ 
Finally, each honest recruited mailman can either keep serving for future senders without requesting withdrawals or stop working and
withdraw $\mathsf{\Xi} deposit$ and accumulated $\mathsf{\Xi} remuneration$ from $C_{agent}$.
% }

\subsection{Limitations and challenges}
\label{s4.C}

% We identify two key limitations of the strawmaforn protocol that are 
We identify two challenges that significantly limit the security and scalability of the strawman protocol.
% are hard to solve:
% We observe that both the two problem were unsolved in previous work~\cite{ning2018keeping,li2018decentralized,Kimono}, thus provoking the generation of our new protocol $SilentDelivery$.

\subsubsection{Premature revelation of recruitment relationships}

% \noindent (1) \textit{\textbf{Premature revelation of recruitment relationships}}:
% We consider the first limitation to be the \textit{premature revelation of recruitment relations}.
In the strawman protocol, a recruitment agreement (or agreement) is first signed by a registered mailman and then signed by a sender. After that, the agreement comes into force at step 2 while at the same moment the included recruitment relationship (or relationship) is made known to all via the uploaded addresses of recruited mailmen, namely $List(addr(M))$.
Later, the protocol ends at step 9, allowing each honest recruited mailman to make withdrawals.
% Later at protocol step 9, a recruited mailman, who did not disclose its maintained $share$ before \textit{TIDS.deliver}, can withdraw both $\mathsf{\Xi} deposit$ and $\mathsf{\Xi} remuneration$.
After revisiting this procedure, we observe that the relationships were getting revealed much earlier than necessary.
Specifically, the relationships were made public at step 2 while only at step 9, namely the settlement stage, it becomes necessary for $C_{agent}$ to know $List(addr(M))$ to approve withdrawals.
% Therefore, the recruitment relationships were getting revealed much earlier than necessary.
Such a premature revelation of relationships endangers the security of \textit{TIDS} in two aspects.

First, it has been widely recognized that the anonymity offered by blockchain networks is not strong~\cite{biryukov2014deanonymisation} and therefore, premature revelation helps an adversary locate recruited mailmen before \textit{TIDS.deliver}. 
% \textcolor{red}{ 
Specifically, for the sake of anonymity, one may create a new Ethereum account to be a mailman.
However, this new account must have been transferred an amount of Ether by an existing account in order to call $newMailman()$ and pay $\mathsf{\Xi} deposit$.
In case that the information of the owner of the existing account has been disclosed at public places such as a forum, the anonymity of the mailman account has been breached via the connection of the two accounts.
An adversary could then leverage the disclosed information about the existing account to attack the new mailman account.
% even in the case that a registered mailman is a newly created Ethereum account, it must have been transfered an amount of Ether by another account for the purpose of calling $newMailman()$ and paying $\mathsf{\Xi} deposit$.
% }
This significantly weakens the underlying protections of the mailmen offered by the large-scale anonymous Ethereum P2P network. 
% \textcolor{red}{ 
In an ideal scenario when no information about the relationships is disclosed, from the view of an adversary, all the registered mailmen have equal probability to be recruited by a sender.
% }
% {\bf a little bit more explanation on why the security is affected will be helpful}

Second, in case of a recruited mailman seeking collusion with any other party, the premature revelation helps the mailman prove its relationship to that party, which further promotes the success of the trading between these mutually distrusted two parties.
What makes the problem more challenging is that the relationships may get revealed through side information, other than just $List(addr(M))$.
For instance, even if we conceal $List(addr(M))$ in step 2, due to the public $List(hash(share))$, a mailman can still prove its relationship by revealing its maintained $share$.
To make matters worse, 
even if the strawman protocol has forbidden any premature disclosure of $shares$ through $reportPremature()$,
it is still possible for a sophisticated mailman to bypass this restriction by offering a zero-knowledge proof $\pi$~\cite{ben2014succinct} to demonstrate that the mailman knows the pre-image of $hash(share)$ (i.e., $share$), without revealing $share$, thus being able to prove its relationship without being panelized.
Therefore, to overcome this difficulty, we need a solution that can conceal both $List(addr(M))$ and all possible side information (e.g., $List(hash(share))$) before \textit{TIDS.deliver} while still making $List(addr(M))$ get back in $C_{agent}$ during \textit{TIDS.deliver} to help $C_{agent}$ process withdrawal requests.

\subsubsection{A tradeoff between scalability and availability}

% \noindent (2) \textit{\textbf{A tradeoff between scalability and availability of \textit{TIDS}}}:
We observe that the non-scalable design of the strawman protocol leads to $O(n)$ gas cost for a sender to recruit $n$ mailman, which makes the protocol hard to scale in practice.
% that makes completion of a service spend gas in scale $O(n)$ when $n$ mailmand are recruited.
There are two places in the strawman protocol that make the gas cost go with $O(n)$.
First, at step 2, after sender $S$ uploads $List(addr(M))$, an amount of gas needs to be spent in $C_{agent}$ to change the state of each recruited mailman and bind this mailman with this service.
Second, at step 6, each recruited mailman needs to spend an amount of gas to reveal its maintained $share$ to $C_{agent}$ so that the $share$ can be verified through the $hash(share)$ in $C_{agent}$ and made known to recipient $R$.
Obviously, the simplest way of reducing gas cost would be recruiting fewer mailmen.
However, given a fixed $t$ regarding $(t,n)$-threshold secret sharing as well as a fixed probability that a single $share$ gets lost, a larger $n$ results in a higher probability of recovering $key$ at protocol step 7, namely higher availability of \textit{TIDS}~\cite{rodrigues2005high}.
% Unfortunately, security offered by the $(t,n)$-threshold secret sharing is highly correlated with the value of $n$. 
% [\textbf{lemma}]
This shows the tradeoff between scalability and availability of the \textit{TIDS}.
To address this, we propose a redesign of the protocol so that the non-scalable regulations within $C_{agent}$ can be removed while the removed regulations can still constrain the behaviors of protocol participants just as if these regulations are still bounded with $C_{agent}$ in blockchain.
% Inspired by the recently advanced techniques for improving blockchain scalability, including payment channel network (PCN)~\cite{dziembowski2017perun} and state channel network (SCN)~\cite{dziembowski2018general}, we proposes a solution that 

% \textcolor{red}{
\subsection{Attack methods}
Two attack methods can be used to disclose the private information before the prescribed time-frame.
Specifically, we assume that there exists an adversary 
seeking the premature disclosure of private information before the future time-frame.
For instance, Alice's competitors may want to learn the content of her business proposal before the board meeting when Alice is using the service to release the proposal exactly during the board meeting time.
It is easy to see that such premature disclosure could happen when an adversary acquires enough shares of the decryption key from the recruited mailmen before the time-frame. The adversary can employ two methods to achieve this goal.
The first attack method is based on subverting the security of the protocol by creating a large number of EOAs leading to a Sybil attack~\cite{douceur2002Sybil} to gain a disproportionately large influence in getting recruited as a mailmen and recovering the decryption key from the shares in the mailmen. The second attack method is based on creating a monetary reward to bribe the recruited mailmen to disclose the key shares prior to the release time. We refer to it as the Bribery attack.

\section{The \textit{SilentDelivery} Protocol}
\label{s5}

% In this section, we present the proposed \textit{SilentDelivery} protocol, which is sketched in Fig.~\ref{protocol_sketch_02} and formally described in Fig.~\ref{protocol_detail_2} and Fig.~\ref{protocol_detail_3}. 
% Specifically, we first present the \textit{TIDS.send} phase and highlight its \textit{silent recruitment} component that is specially designed to handle the first main limitation of the strawman protocol.
% We then depicts the \textit{TIDS.pend} and \textit{TIDS.deliver} phases, which together forms the \textit{dual-mode execution} component for resolving the second main limitation of the strawman protocol.

% In this section, we start from analying the two key limitations of the strawman protocol and
% sketching their countermeasures in the proposed \textit{SilentDelivery} protocol.
In this section, we start by introducing our key ideas for tackling the challenges that limit the strawman protocol.
% , presenting its key design ideas and highlighting the two components that are specially designed for overcoming the two key limitations discussed in section~\ref{s3.4}.
We then present the proposed \textit{SilentDelivery} protocol and  focus on its two novel countermeasure strategies, namely \textit{silent recruitment} and \textit{dual-mode execution}.
% The proposed \textit{SilentDelivery} protocol is sketched in Fig.~\ref{protocol_sketch_02} and formally described in both Fig.~\ref{protocol_detail_2} and Fig.~\ref{protocol_detail_3}. 
We sketch the \textit{SilentDelivery} protocol in Fig.~\ref{protocol_sketch_02} and formally describe it in Fig.~\ref{protocol_detail_2} and Fig.~\ref{protocol_detail_3}.

\subsection{Tackling the challenges 
% {\bf it will be good to mention what challenges are tackled in the subsection title. E.g., tackling challenge name 1 and challenge name 2} \textcolor{red}{ 
% This subsection offers a comprehensive discussion of both the two challenges}
}
\label{s4.1}

\begin{figure}
\centering
{
   
    \includegraphics[width=9cm,height=9.5cm]{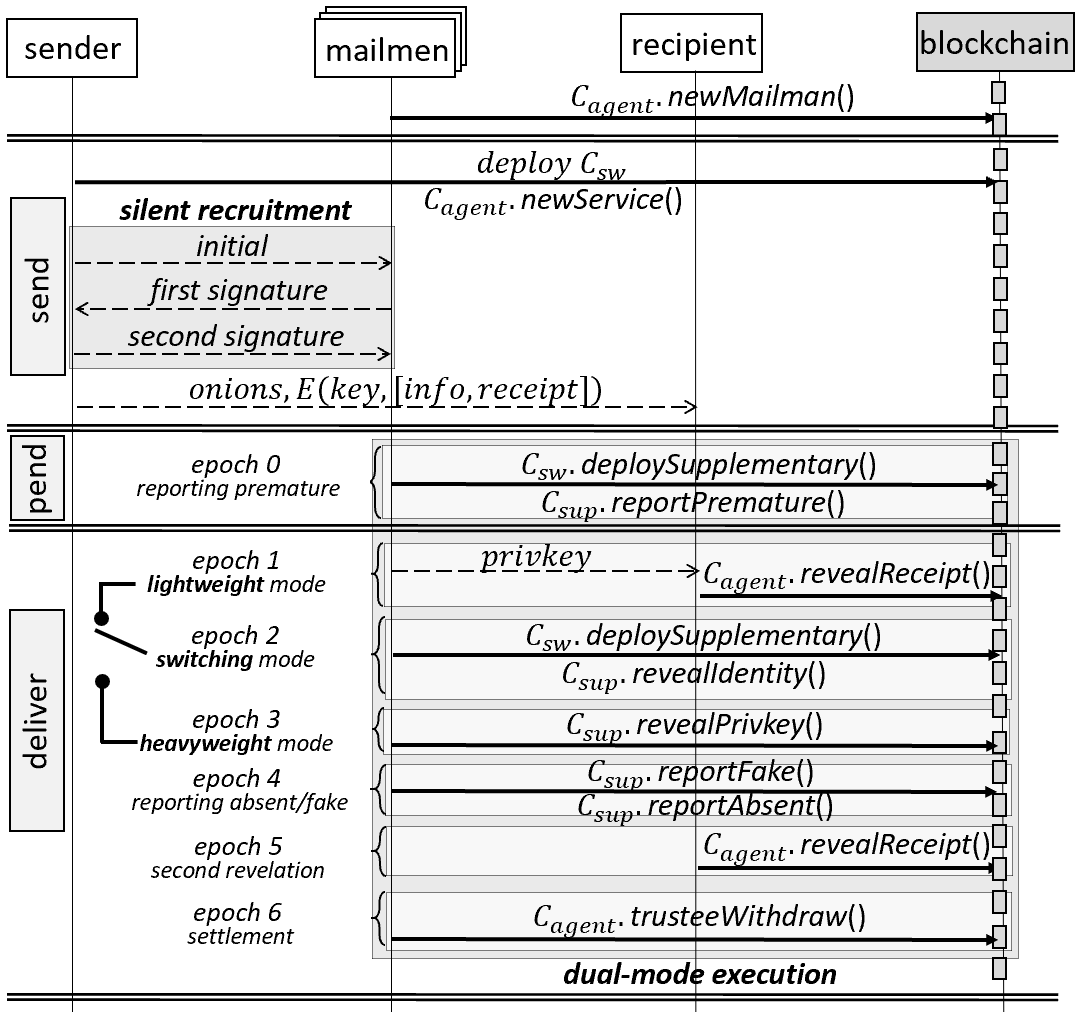}
}
%\vspace{-6mm}
\caption {\textit{SilentDelivery} protocol sketch}
%\vspace{-5mm}
\label{protocol_sketch_02} 
\end{figure}

We observe both similarities and differences in tackling the two challenges presented in Section~\ref{s4.C}.
Regarding similarities, both the two challenges demand a way for removing data from the agent contract $C_{agent}$ while they also need the removal to get rolled back when a certain condition is reached.
Specifically, a solution for the first challenge, namely the premature revelation of recruitment relationships, should remove recruitment relationships from $C_{agent}$ before \textit{TIDS.deliver} and later reveal the relationships to $C_{agent}$ during \textit{TIDS.deliver}.
Similarly, a solution for the second challenge, namely the tradeoff between scalability and availability of \textit{TIDS}, should remove non-scalable regulations from $C_{agent}$ and later reveal these regulations to $C_{agent}$ in case of any violation of them occurred.
% We observe that the potential solutions for resolving the two challenging limitations of the strawman protocol may have some properties in common.
% The first solution is expected to conceal any information that may reveal recruitment relationships from $C_{agent}$ while the second solution is demanded to remove non-scalable regulations from $C_{agent}$.
% Nevertheless, when the time goes to the range of \textit{TIDS.deliver}, the hidden recruitment relationships needs to be revealed to $C_{agent}$ so that the settlement stage can get started.
% Similarly, when any participant of the protocol violates the removed regulations, these removed regulations are expected to be re-binded with $C_{agent}$ to penalize the dishonest participants.
% To sum up, they both need on-chain data to get removed for the purpose of enhancing security or scalability. However, when a certain condition is reached, they also need the removed data to be turned back on-chain.
The main differences between them are twofold:
(1) the removed relationships should be concealed against any third party while the removed regulations should still be made public;
(2) the rollback of the removal of relationships is inevitably driven by the arrival of time-frame while the rollback of the removal of regulations is driven by violations of these regulations, which is completely avoidable when the recruited mailmen are honest.

By keeping these similarities and differences in mind, in \textit{SilentDelivery}, we design two different countermeasure strategies for tackling the two challenges.
As illustrated in Fig.~\ref{protocol_sketch_02}, the first countermeasure strategy is implemented as a \textit{silent recruitment} component within the \textit{TIDS.send} phase.
The key idea behind it is to consider the relationships also a type of private information that demands timed delivery, so both senders and mailmen should sign their recruitment agreements secretly and their signatures should also be protected with the $key$ split into $shares$. 
The second countermeasure strategy, namely \textit{dual-mode execution}, allows the protocol to be executed in two different modes.
% The second countermeasure strategy is inspired by the recently advanced techniques for improving blockchain scalability, represented by payment channel network (PCN)~\cite{dziembowski2019perun} and state channel network (SCN)~\cite{dziembowski2018general}.
%  we redesign the protocol with two modes of completing the service (i.e., \textit{dual-mode execution} in Fig.~\ref{protocol_sketch_02},), namely a lightweight mode and a heavyweight mode. 
Specifically, we cut the non-scalable regulations out of $C_{agent}$ and re-organize them as a supplementary contract $C_{sup}$. 
By default, when recruited mailmen honestly follow the removed regulations, the protocol goes with its lightweight mode without involving any $O(n)$ on-chain interactions.
If any recruited mailman violates the removed regulations, any honest mailman can rebind the removed regulations (i.e., $C_{sup}$) with $C_{agent}$ to penalize the violations, which turns the protocol into the heavyweight mode and results in a pay-cut.
Thus, the penalty of misbehaviors and the pay-cut induced by the heavyweight mode can incentivize recruited mailmen to stay honest, making the protocol stays at its lightweight mode with $O(1)$ gas cost.

% a \textit{dual-mode execution} process spreading over seven epochs.

\begin{figure}
\begin{minipage}{0.5\textwidth}
\begin{small}
\begin{mdframed}[innerleftmargin=8pt]

% \centerline{\textbf{Silent recruitment}}

% \noindent \textbf{Sender [setup]:} 
\begin{packed_enum}[leftmargin=*]
  \item Sender $S$ deploys a switch contract $C_{sw}$ and obtains both $addr(C_{sw})$ and $addr(C_{sup})$.
  % \citem Sender $S$ computes address of supplementary contract: \\
  % $addr(C_{sup}) \gets computeadd(addr(C_{sw}))$.
  \item 
  Sender $S \Rightarrow C_{agent}.newService([day,slot],[l,t,n],\\
  addr(C_{sw}),addr(C_{sup}),addr(R),hash(receipt))$.
  % \item 
  % Recipient $R \Rightarrow C_{agent}.recipientDeposit(addr(S),\\
  % addr(C_{sw}))$.
\end{packed_enum}

% \begin{mdframed}[innerleftmargin=8pt]

\begin{center}
%\vspace{-1mm}
\rule{25mm}{0.3mm}
\textbf{Silent recruitment}
\rule{25mm}{0.3mm}
%\vspace{-1mm}
\end{center} 

\noindent \textbf{Sender [initial]:} 
\begin{packed_enum}[leftmargin=*]
  \setcounter{enumi}{2}
  \citem $\dashrightarrow M$: $index$, $addr(C_{sw})$, $C_{sup}$, \\
  $vrs_{sup} \gets sig(hash(addr(C_{sw}), C_{sup}))$.
\end{packed_enum}
\noindent \textbf{Mailman [first signature]:} 
\begin{packed_enum}[leftmargin=*]
  \setcounter{enumi}{3}
  \citem Verify service information with $C_{agent}$ and $C_{sup}$.
  \citem $\dashrightarrow S$: 
  $vrs_{m} \gets sig(hash(addr(C_{sw}),index))$.
\end{packed_enum}
\noindent \textbf{Sender [second signature]:} 
\begin{packed_enum}[leftmargin=*]
  \setcounter{enumi}{5}
  \citem $vrs_{s} \gets sig(hash(addr(C_{sw}),index,vrs_{m}))$.
  \citem Create a 256-bit secret $key$ as well as a 256-bit secret $receipt$ and obtains $shares \gets SS.split(key,(t,n))$.
  % \citem Create a 256-bit secret $key$ and a $receipt$.
  % \citem Split $key$ to $n$ $shares$ through $(m,n)$ threshold secret sharing:
  % $shares \gets SS.split(key,(m,n))$.
  \citem Encrypt $shares$ to $onions$ with $pubkeys$ of selected mailmen:
  $onions \gets E(pubkeys,shares)$.
  \citem $\dashrightarrow$ each $M$: 
  $List(E(key,[index,vrs_s,vrs_m]))$, $onions$,\\
  $vrs_{sm} \gets sig(hash(List(E(key,[index,vrs_s,vrs_m])),$ \\
  $onions))$.
  % \citem $\dashrightarrow R$: 
  % $E(key,[secret, receipt])$, $onions$.
\end{packed_enum}

\begin{center}
%\vspace{-2mm}
\rule{75mm}{0.3mm}
%\vspace{-1mm}
\end{center} 

\begin{packed_enum}[leftmargin=*]
  \setcounter{enumi}{9}
  \citem Sender $S$ $\dashrightarrow R$: 
  $E(key,[info, receipt])$, $onions$, \\
  $vrs_{st} \gets sig(hash(E(key,[info, receipt]), onions))$.
\end{packed_enum}

\end{mdframed}
\end{small}
\end{minipage}

\captionof{figure}{\textit{TIDS.send} in \textit{SilentDelivery}}
%\vspace{-5mm}
\label{protocol_detail_2}
\end{figure}

\begin{figure*}
  \begin{minipage}{\textwidth}
  \begin{small}
\begin{mdframed}[innerleftmargin=8pt]

% \centerline{\textbf{Lightweight-first Enforceable Execution}}

\noindent \textbf{Epoch-0 [reporting premature]:}
\begin{packed_enum}[leftmargin=*]
  \item A mailman $M$ does the folowing and the protocol \underline{jumps to epoch-2}.
  \begin{packed_enum}[leftmargin=*]
    \item 
    % Deploy the supplementary contract $C_{sup}$ through the existed switch contract $C_{sw}$, thus switching the protocol to the heavy mode: 
    $\Rightarrow C_{sw}.deploySupplementary(addr(C_{sw}),C_{sup},vrs_{sup})$.
    \item 
    % Report the premature disclosure through 
    $\Rightarrow C_{sup}.reportPremature(index,privkey)$. 
  \end{packed_enum}
\end{packed_enum}

\noindent \textbf{Epoch-1 [lightweight mode]:} 
\begin{packed_enum}[leftmargin=*]
  \setcounter{enumi}{1}
  \citem Each mailman $M$ reveals its $privkey$.
  \citem Recipient $R$ gets $shares \gets D(privkeys,onions)$.\
  \item If $|shares|>=t$, recipient $R$ does the following and the protocol \underline{jumps to epoch-6}.
  \begin{packed_enum}[leftmargin=*]
    \item Restore $key \gets SS.restore(shares,[t,n])$ and obtain $[info,receipt] \gets D(key,[info, receipt])$.
    \item $\Rightarrow C_{agent}$: 
  $recipientReceipt(receipt,addr(S),addr(C_{sw}))$.
  \end{packed_enum}
  Otherwise, the protocol \underline{goes to epoch-2.}
\end{packed_enum}

\noindent \textbf{Epoch-2 [switching mode]:} 
\begin{packed_enum}[leftmargin=*]
  \setcounter{enumi}{4}
  \item 
  % In case that the supplementary contract $C_{sup}$ has not yet been deployed during epoch-0, any mailman $M$ could deploy $C_{sup}$ through the existed switch contract $C_{sw}$, thus switching the protocol to the heavy mode: \\
  % A mailman $M$ switches mode: 
  $M \Rightarrow C_{sw}.deploySupplementary(addr(C_{sw}),C_{sup},vrs_{sup})$.
  \citem Mailmen reveal their maintained $privkeys$ to public and collect $privkeys$ from others. 
  \item If $|shares|>=t$, mailmen should do the following, which makes the protocol \underline{enter epoch-3}.
  \begin{packed_enum}[leftmargin=*]
    \item Compute $key \gets SS.restore(shares,(t,n))$ and $List(index,vrs_s,vrs_m) \gets List(D(key,[index,vrs_s,vrs_m]))$.
    \item 
    % Reveal identities of all mailmen through 
    $\Rightarrow C_{sup}.revealIdentity(List(index,vrs_s,vrs_m))$.
  \end{packed_enum}
  Otherwise, the protocol \underline{jumps to epoch-6}.
\end{packed_enum}

\noindent \textbf{Epoch-3 [heavyweight mode]:} 
\begin{packed_enum}[leftmargin=*]
  \setcounter{enumi}{7}
  \item 
  % After identities of all mailmen have been revealed to $C_{sup}$, each mailman $M$ should reveal its $privkey$ corresponding to the $[day,slot]$ selected by sender $S$ through 
  $M \Rightarrow C_{sup}.revealPrivkey(index,privkey)$. \\
  At the end of epoch-3, the protocol \underline{goes to epoch-4}.
\end{packed_enum}

\noindent \textbf{Epoch-4 [reporting absent/fake]:} 
\begin{packed_enum}[leftmargin=*]
  \setcounter{enumi}{8}
  \item 
  % Upon detecting an absent mailman who did not reveal its $privkey$ during epoch-3, any mailman could report this misbehavior through
  $M \Rightarrow C_{sup}.reportAbsent(index)$. 
  \item 
  % Upon detecting a mailman who revealed a fake $privkey$ during epoch-3, any mailman could report this misbehavior through
  $M \Rightarrow C_{sup}.reportFake(index)$. \\
  % \item After reporting any misbehavior, a mailman could transfer the record of this misbehavior from $C_{sup}$ to $C_{agent}$ through 
  % $\Rightarrow C_{sup}.informAgent(index)$. \\
  At the end of epoch-4, the protocol \underline{goes to epoch-5}.
\end{packed_enum}

\noindent \textbf{Epoch-5 [second revelation]:} 
\begin{packed_enum}[leftmargin=*]
  \setcounter{enumi}{11}
  \item Recipient submits the $receipt$ through 
  $\Rightarrow C_{agent}.recipientReceipt(receipt,addr(S),addr(C_{sw}))$. \\
  At the end of epoch-5, the protocol \underline{goes to epoch-6}.
\end{packed_enum}

\noindent \textbf{Epoch-6 [settlement]:} 
\begin{packed_enum}[leftmargin=*]
  \setcounter{enumi}{12}
  \item Depending on how the protocol terminates, the protocol participants can request withdrawals at $C_{agent}$.
\end{packed_enum}

\end{mdframed}
\end{small}
\end{minipage}

\captionof{figure}{\textit{TIDS.pend} and \textit{TIDS.deliver} in \textit{SilentDilivery}. They together form the \textit{dual-mode execution}.}
%\vspace{-6.5mm}
\label{protocol_detail_3}
\end{figure*}

\subsection{Silent recruitment}
\label{s4.2}

Before presenting the \textit{TIDS.send} phase, it is worth noting that \textit{SilentDilivery} demands each mailman to create a list of $[privkey,pubkey]$ key pairs for a list of future time-frames $[day,slot]$, maintain all the $privkeys$ by itself and submit all the $pubkeys$ to the agent contract $C_{agent}$ during registration via $newMailman()$.
A private key in a dedicated key pair is designed to be revealed by the mailman at a prescribed time point. 
Besides the onion planned to be decrypted, there will be no other information protected by the private key.
In other words, the private key is designed to never be reused.
During \textit{TIDS.send}, sender $S$ first deploys a switch contract $C_{sw}$, which contains a single function \textit{deploySupplementary()}.
Just like the name implies, through $C_{sw}$, any honest mailman can deploy the supplementary contract $C_{sup}$ and turn the protocol from the default lightweight mode into the heavyweight mode.
Since the only transaction that is allowed to be sent by $C_{sw}$ creates $C_{sup}$, the address of $C_{sup}$ is deterministic and can be computed by sender $S$ at protocol step 1.
After that, sender $S$ selects mailmen and sets up a service with $C_{agent}$ without disclosing any (side) information revealing recruitment relationships.
% \textcolor{red}{
Similar to the selection of mailmen in the strawman protocol, the selection of mailmen in the Silent Recruitment is also random. This is a straightforward way of defending against 
adversaries that try to leverage a few accounts to possess the majority of mailmen recruited by a particular sender.
% As can be seen, the service inputs contain no information that may disclose recruitment relations.
% Then, \textit{SilentDilivery} also demands recipient $R$ to lock a $\mathsf{\Xi} deposit)$ in $C_{agent}$ because a dishonest recipient may deliberately
From then on, \textit{silent recruitment} is executed via private channels in the form of a three-way handshake.
Specifically, sender $S$ initials the handshake by giving each mailman the two contracts $C_{sw}$ and $C_{sup}$, a signature $vrs_{sup}$ regarding the two contracts and an $index$ assigned to the mailman.
% details of $C_{sw}$ and $C_{sup}$ as well as signature $vrs_{sup}$ regarding the two contracts to a group of mailmen randomly selected from the registered one.
Upon getting contacted, each mailman verifies the correctness of service information and sends back a signature $vrs_{m} \gets sig(hash(addr(C_{sw}),index))$ to sender, which implies that the mailman has agreed to take charge of this service.
% take $index$ of the service corresponding with $addr(C_{sw})$.
% on providing the service.
Upon receiving mailman's signature, sender $S$ also generates a signature $vrs_{s} \gets sig(hash(addr(C_{sw}),index,vrs_{m}))$, which says that sender $S$ has agreed to recruit the signer of $vrs_{m}$ in this service.
% the service corresponding with $addr(C_{sw})$.
Then, similar to the strawman protocol, sender $S$ creates $shares$ of a $key$ and a $receipt$.
However, unlike the strawman protocol, these $shares$ are not directly given to mailmen. Instead, at protocol step 8, each $share$ is iteratively encrypted with $l$ $pubkeys$ from $l$ different recruited mailmen, where $l$ is a parameter determined by sender $S$ at protocol step 2.
In this way, each $share$ is turned into an $onion$ and its recovery needs $privkeys$ maintained by $l$ mailmen.
This design allows the premature disclosure of a $share$ to be verified through pairing the disclosed $privkeys$ with the $pubkeys$ in $C_{agent}$, instead of having to rely on $hash(share)$ that reveals recruitment relationships.
Finally, through private channels, sender $S$ broadcasts all these $onions$, transmit a list of encrypted tuple $[index,vrs_s,vrs_m]$ to each recruited mailman and encrypted $[info,receipt]$ to the recipient.
% \textcolor{red}{
With the signatures $vrs_{sm}$ and $vrs_{st}$, both mailmen and recipient will be able to verify the received messages and request the sender to resend the messages if needed.

\subsection{Dual-mode execution}
\label{s4.3}
%\vspace{-1mm}

Next, we present the \textit{dual-mode execution}, which incentivizes mailmen to make the protocol get executed in the lightweight mode, reducing the service cost from $O(n)$ to $O(1)$.
% By default, \textit{SilentDilivery} is executed in the lightweight mode, resulting in $O(1)$ gas cost by assuming all recruited mailmen are honest.
% In case of any misbehaviors, the protocol is switched to the heavyweight mode, which spends gas in $O(n)$ and reduces remuneration received by recruited mailmen. Thus it incentivizes the mailmen to stay honest to earn a higher profit.  
% which incentivizes all recruited mailmen to stay honest to make the protocol get executed in $O(1)$ lightweight mode, instead of performing any misbehavior to switch the protocol into $O(n)$ heavyweight mode and reduce their remuneration.
% allowing the protocol to be executed by default in $O(1)$ lightweight mode
% while getting switched to $O(n)$ heavyweight mode in case of any misbehavior occurred.
% our new design, dual-mode execution, that turns the heavyweight execution in the strawman protocol into the lightweight-default dual-mode execution in \textit{SilentDilivery}, which consists of seven epochs that spread across \textit{TIDS.pend} and \textit{TIDS.deliver}.
% from a heavyweight single-mote execution into a lightweight-first dual-mode execution, which consists of seven epochs that spread across \textit{TIDS.pend} and \textit{TIDS.deliver}.
We design the \textit{dual-mode execution} to include six epochs spreading across the \textit{TIDS.deliver} phase.
Specially, we consider the \textit{TIDS.pend} phase as \textit{epoch-0}, during which any mailman can switch the default lightweight mode to the heavyweight mode by deploying $C_{sup}$ via $C_{sw}$ and report a prematurely disclosed $privkey$.
% \textcolor{red}{
Specifically, the protocol allows any mailman, who knows the identity of a mailman recruited by a particular sender and also obtains the corresponding $privkey$ used by that mailman to encrypt the onion for that sender, to report a premature disclosure.~\footnote{
\begin{scriptsize} 
% \textcolor{red}{
Any party obtaining the $privkey$, even if not controlling a mailman, can fairly cooperate with any mailman via smart contracts to report a premature disclosure. 
Concretely, the party (say party A) can claim the ownership of a disclosed $privkey$ by creating a trading smart contract to make a deal with another party (say party B) who owns a mailman account. 
Then, if party B wants to obtain this disclosed $privkey$, party B should lock an amount of money (say one third of the amount of security deposit) to the trading smart contract.
Finally, party A needs to disclose $privkey$ to the trading smart contract, which triggers the trading contract to verify two things, namely the correctness of the $privkey$ using $pubkey$ and the status of supplementary contract $C_{sup}$ (i.e., whether $C_{sup}$ has not yet been deployed). 
After these verifications, the trading contract will automatically deploy the supplementary contract $C_{sup}$ on behalf of party B and transfer the locked money to party A.
As a result, party B will become the reporter who has deployed the supplementary contract $C_{sup}$ and both the two parties will receive monetary rewards.
\end{scriptsize}}~\footnote{
\begin{scriptsize} 
% \textcolor{red}{
In the case that two reporters simultaneously send out the reporting transaction and the two transactions happen to be grouped in the same block by a blockchain miner, only the first transaction executed by the miner could set the transaction creator to be the reporter recorded in the contract $C_{agent}$. The second transaction will be rejected.
\end{scriptsize}}
Later, if the reported $privkey$ is proved to be correct or the suspected mailman becomes absent in \textit{epoch-4}, the reporter could receive a monetary reward from the security deposit paid by the suspected mailman as the incentive, and the rest of the confiscated security deposit will be used to compensate for the cost of mode switching.
In contrast, if the reported $privkey$ is proved to be incorrect after the suspected mailman has revealed the corresponding $privkey$ in \textit{epoch-4}, the reporter who is also a mailman, will lose the locked security deposit. The security deposit will be used to compensate for the cost of mode switching.
Then, in \textit{epoch-1}, if the recruited mailmen honestly obey the non-scalable regulations shifted from $C_{agent}$ to $C_{sup}$ by revealing their $privkeys$ to recipient $R$ via private channels, $R$ will be able to recover $key$ and acquire $info$ and $receipt$, which makes the protocol reach a successful termination 
% in lightweight mode without any $O(n)$ on-chain operation 
and jump to \textit{epoch-6}, namely the settlement stage.
In case that some mailmen violate the regulations in $C_{sup}$ and result in a failure of recovering $key$, the protocol will enter \textit{epoch-2}, during which any honest mailman can switch the protocol into the heavyweight mode (i.e., \textit{epoch-3}) by deploying $C_{sup}$ via $C_{sw}$. 
After that, the recruited mailmen reveal $privkeys$ to all via Whisper protocol, decrypt $onions$ to $shares$,  recover $key$ from $shares$ and finally acquire the list of tuple $[index,vrs_s,vrs_m]$. 
Each tuple can prove a recruitment relationship to $C_{sup}$ by revealing an agreement signed by a mailman via $vrs_{m}$ and by a sender via $vrs_{s}$.
If the $key$ can not get recovered, the protocol researches a failed termination and jumps to \textit{epoch-6}.
During \textit{epoch-3}, the protocol is executed in the heavyweight mode, demanding each recruited mailman to reveal its $privkey$ to $C_{sup}$ via $revealPrivkey()$.
Furthermore, any absent or fake $privkey$ that were not appropriately submitted during \textit{epoch-3} will be identified as misbehavior and reported during \textit{epoch-4}, resulting in the dishonest mailman to lose $\mathsf{\Xi} deposit$.
After that, \textit{epoch-5} provides the second chance of making the protocol end with success, though in heavyweight mode.
% \textcolor{red}{
Finally, during \textit{epoch-6}, if the service is failed, each mailman could withdraw its deposit and the sender could withdraw remained fee anytime after epoch-5. 
Otherwise, the service is successful. Each mailman, with no misbehavior, could withdraw its deposit, remuneration or reporting award. In case that the protocol terminates via lightweight mode, mailmen need to follow step 6 and 7.1 to prove relationships to request withdrawals. The sender could withdraw remained fee anytime after epoch-5.

% cut contract to agent and supplementary\\
% agent on-chain while supplementary off-chain\\
% supplementary everyone know/off-chain before \\
% subject: supplementary\\
% line: dispute [not firm]\\
% before: off-chain, mailman knows\\
% after: on-chain

% The first one demands a solution that can offload the recruitment identifiers from blockchain before the time-frame while uploading the identifiers back blockchain during the time-frame.
% The second challenge demands a solution that can offload the non-scalable regulations from $C_{agent}$ when recruited mailmen are self-disciplined while uploading these regulations back blockchain in case of a dispute occurred.
% To sum up, they both need on-chain data to get offloaded for the purpose of enhancing security or reducing cost.
% However, when a certain condition is reached, they also need the offloaded data to be turned back on-chain.
% In other words, the recovery of offloaded data is time-driven in the first case while event-driven in the second case.

% As illustrated by Fig.~\ref{protocol_sketch_02}, the first challenge is mainly handled in \textit{TIDS.send} through a process named \textit{Silent recruitment}, which is designed as a three-way handshake process.
% We design

% Specifically, the non-scalable regulations are cut out of $C_{agent}$ and re-organized as a supplementary contract $C_{sup}$, which is signed by sender and maintained by all mailmen.
% If all recruited mailmen are honest, there will be no need for $C_{sup}$ to be re-deployed.
% When a dispute occurs, 

\section{Security and cost analysis}

We next analyze the security of \textit{SilentDelivery} in terms of its termination, attack resilience and availability properties.
After that, we analyze the cost of \textit{SilentDelivery}.

\subsection{Termination}
% \textcolor{red}{
A fair trading should guarantee that honest service requestors only need to pay service fee when the prescribed service is successful and honest service providers always receive service fee when the prescribed service is successful.
As illustrated by Fig.~\ref{protocol_detail_3}, the protocol can always terminate by reaching \textit{epoch-6}.
During the settlement stage, fairness is guaranteed in the following ways:\\
\noindent 
(1) The $\mathsf{\Xi} remunaration$ is only charged from sender $S$ when the service is successful and is only paid to honest mailmen who have no detected misbehavior.\\
\noindent 
(2) As described in protocol step 1, 9 and 10, $\mathsf{\Xi} deposit$ is confiscated only when a recruited mailman fails to protect its $privkey$ before \textit{TIDS.deliver} or fails to properly reveal $privkey$ during \textit{TIDS.deliver}. In other words, since the $privkeys$ are necessary for reporting recruited mailmen and are only known by the recruited mailmen, it becomes difficult for any party, including senders and recipients, to unjustly accuse a mailman by abusing the reporting functions.
This property was not achieved in the strawman protocol, where a dishonest sender always owns the ability to swindle $\mathsf{\Xi} deposit$ of a recruited mailman because the $share$ reported through $reportPremature(share)$ is not uniquely known by the mailman, but also by the sender.\\
% (3) As illustrated by protocol step  

% report the premature disclosure of $share$ 

% $\mathsf{\Xi} deposit$ is confiscated only when its a recruited mailman violates the protocol

% the protocol guarantees that 

% (1) $\mathsf{\Xi} remunaration$ should be and only be charged if the service is successful and 
% (2) $\mathsf{\Xi} deposit$ should be and only be confiscated if its owner really violates the protocol.

\subsection{Attack resilience}
For analyzing the attack resilience, we assume that there exists an adversary 
seeking the premature disclosure of private information $info$ before the future time-frame.
% aiming at attacking a specific sender $S$ by obtaining the private information $info$ before the future time-frame.
% For instance, Alice's competitors may want to learn the content of her business proposal before the board meeting when Alice is using the service to release the proposal exactly during the board meeting time.
% while the rest of EOAs are rational adversaries. The malicious adversary may choose to launch either a time difference attack or an execution failure attack. 
There are two approaches to launch such an attack, namely mailman Bribery attack~\cite{dong2017betrayal} and Sybil attack~\cite{douceur2002Sybil}. 
% \textcolor{red}{
Before presenting the details, we denote the threshold of secret sharing as $t$, the layer of onion as $l$, the security deposit charged to a single mailman as $d$, the number of innocent mailmen as $v$, the number of malicious mailmen as $x$ and the percentage of malicious mailmen as $p_M$.

% Through mailman bribery, the adversary can deploy a smart contract with a fund larger than $\mathsf{\Xi} deposit$ and use this smart contract as bait to bribe a recruited mailman.
% , even if the mailman's identity is not known. 
% For example, to obtain a specific mailman's $privkey$, the bribery contract can be set with a condition `If any Ethereum account can submit a $privkey$ pairing a specific $pubkey$ in $C_{agent}$, the account can withdraw the fund locked in bribery contract.' 

% \noindent 
%\textcolor{red}{\textbf{Bribery attack:}
The \textit{silent recruitment} realized in \textit{SilentDelivery} makes it difficult for an adversary to succeed via bribery as there is no information regarding the recruitment relationships that can help the adversary distinguish whether a mailman has been recruited or not.%}

%\textcolor{red}{
However, a powerful adversary may identify the identities of some mailmen recruited by a particular sender by launching side-channel attacks.
Then, the adversary could deploy a smart contract with a fund larger than the security deposit $d$ and use this smart contract as bait to bribe a mailman to obtain a specific private key.
Since the fund in the bribery contract is larger than the security deposit, a rational mailman may choose to reveal the private key to the bribery contract to increase its profit.
We now analyze the cost to successfully acquire the private information before the prescribed time-frame via bribery attack.%}

\begin{lemma}
\label{lemma1}
{An adversary needs to spend an amount of $tld$ to bribe mailmen for successfully obtaining the private information before the prescribed time-frame, where $t$ denotes the threshold, $l$ denotes the layer of onion and $d$ denotes the security deposit charged to a single mailman.}
\end{lemma}

\begin{proof}
{An adversary should aim at restoring $key$ before the prescribed time-frame, which means at least $t$ $shares$ should be obtained. To obtain a single $share$, the adversary needs to deploy $l$ bribery smart contracts to collect private keys from $l$ different mailmen, which, due to the existence of the reporting premature mechanism, will cost $ld$. Therefore, the cost of obtaining $t$ $shares$ will be $tld$.}
\end{proof}

% \noindent 
With Sybil attacks~\cite{douceur2002Sybil},
an adversary can create an arbitrary amount of Ethereum accounts and register all these accounts to $C_{agent}$ as mailmen.

\begin{lemma}
\label{lemma2}
% It is infeasible for the premature disclosure of $info$ to succeed through Sybil attacks because 
The minimum amount of security deposit that an adversary needs to deposit to obtain $t$ $shares$ on average via controlled mailmen
would be $(l-1)vd$, where $v$ denotes the number of registered mailmen not controlled by the adversary and $d$ denotes the security deposit charged to a single mailman.
\end{lemma}

\begin{proof}
%The malicious adversary may create as many EOAs as he wants
%{]bf Explain Sybil attack and the reference ~\cite{douceur2002Sybil} in the proof.}
The situation refers to step 2 of the \textit{TIDS.send} phase in \textit{SilentDelivery} protocol, where the registered mailmen can be divided into two groups, namely the (malicious) ones controlled by an adversary through Sybil attacks and the (innocent) ones not controlled by the adversary.
By denoting the number of innocent mailmen as $v$ and the number of malicious mailmen as $x$, we get 
\begin{equation}
p_M=\frac{x}{x+v} \to x=\frac{vp_M}{1-p_M}
\end{equation}
% $$p_M=\frac{x}{x+v} \to x=\frac{vp_M}{1-p_M}$$ 
where $p_M$ denotes the percentage of malicious mailmen.
To obtain a single $share$, all the $l$ mailmen providing $privkeys$ for recovering this $share$ from an $onion$ should be selected from the malicious group, giving its probability $p_M^{l}$. Since there are $n$ shares in total, the overall process can be viewed as a Binomial distribution $B(n,p_M^{l})$ with mean $np_M^{l}$. 
Then, the amount of security deposit $\widehat{d}$ that should be deposited by the adversary to obtain $M$ $shares$ on average would be:
% make $np_M^{l}=t$ would be
% $$\frac{\widehat{d}}{xd}=\frac{t}{np_M^{l}}$$
% , which makes
\begin{equation}
\widehat{d}=\frac{xd}{np_M^{l}} \cdot t=\frac{vp_M}{1-p_M} \cdot \frac{dt}{np_M^{l}} = \frac{vdt}{n} \cdot \frac{p_M^{1-l}}{1-p_M}
\end{equation}
% $$\widehat{d}=\frac{xd}{np_M^{l}} \cdot t=\frac{vp_M}{1-p_M} \cdot \frac{dt}{np_M^{l}} = \frac{vdt}{n} \cdot \frac{p_M^{1-l}}{1-p_M}$$
% .
Since the adversary cannot control $(v,d,t,n)$, to minimize $\widehat{d}$, by setting $\frac{\partial \widehat{d}}{\partial p_M}=0$: 
% $$f(p_M)=\frac{vdm}{n} \cdot \frac{p_M^{2-l}}{1-p_M}$$ and compute 
% $$f'(\widehat{d})=0$$
% , which gives
\begin{equation}
\frac{(1-l)p_M^{-l}}{1-p_M} + \frac{p_M^{1-l}}{(1-p_M)^2}=0 \to p_M=\frac{l-1}{l}
\end{equation}
% $$\frac{(1-l)p_M^{-l}}{1-p_M} + \frac{p_M^{1-l}}{(1-p_M)^2}=0 \to p_M=\frac{l-1}{l}$$
Therefore, when $\widehat{d}$ is minimized:
\begin{equation}
x=v \cdot \frac{l-1}{l} \cdot l = (l-1)v \to \widehat{d}_{min}=(l-1)vd
\end{equation}
% $$x=v \cdot \frac{l-1}{l} \cdot l = (l-1)v \to \widehat{d}_{min}=(l-1)vd$$
\end{proof}
% In Figure~\ref{hipc_0304}, we present the computed schedule success rate when 50\% of mailmen are controlled by a malicious adversary who aims at launching a time difference attack. As can be seen, smaller $m$ while larger $n$ and $l$ help enhancing the resistance against time difference attacks performed through Sybil attack.

% For example, when $v=2500$, $l=5$ and $d=\$200$, $\widehat{d}_{min}$ will be two million dollars.
% \textcolor{red}{
We emphasize that it is hard to entirely prevent bribery attack and Sybil attack, especially in permissionless blockchains.
In this paper, we assume that an adversary has a bounded ability, which refers to a bounded amount of money to be invested in a specific attack. 
Hence, we can quantify the attack resistance with the amount of money that an adversary needs to spend.
% In practice, $d$ and $v$ are system parameters and a sender can customize $l$ based on desirable Sybil attack resistance $\widehat{d}_{min}$.
% }

% \textcolor{red}{
The resilience against Sybil attack depends on three factors. Besides $l$, the two other factors are $v$, the number of registered mailmen not controlled by the adversary and $d$, the security deposit charged to a single mailman. 
In other words, $resilience \propto (l-1)vd$, as in lemma 2.
In practice, a sender can customize $l$ based on desirable Sybil attack resistance $\widehat{d}_{min}$.
Besides, to make the value of $l$ in a reasonable range, we may adjust the other two factors, namely $v$ and $d$. 
It may not be feasible to control $v$. 
However, it is possible to periodically adjust the value of $d$ based on the recent number of registered mailmen and a default value of $l$ that is not large (say $l = 5$).
This is similar to periodically adjusting the mining difficulty in PoW blockchains such as the Bitcoin.

\begin{figure}
\centering
%%%\vspace{-3 mm}
\subfigure[{\small $n=5$}]
{
   \label{hipc_01}
   \includegraphics[width=0.75\columnwidth]{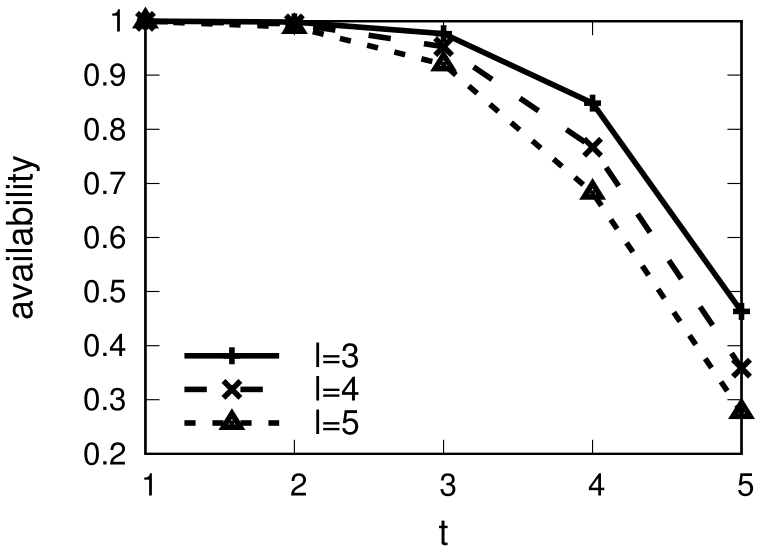}
}
\subfigure[{\small $n=10$}]
{
  \label{hipc_02}
    \includegraphics[width=0.75\columnwidth]{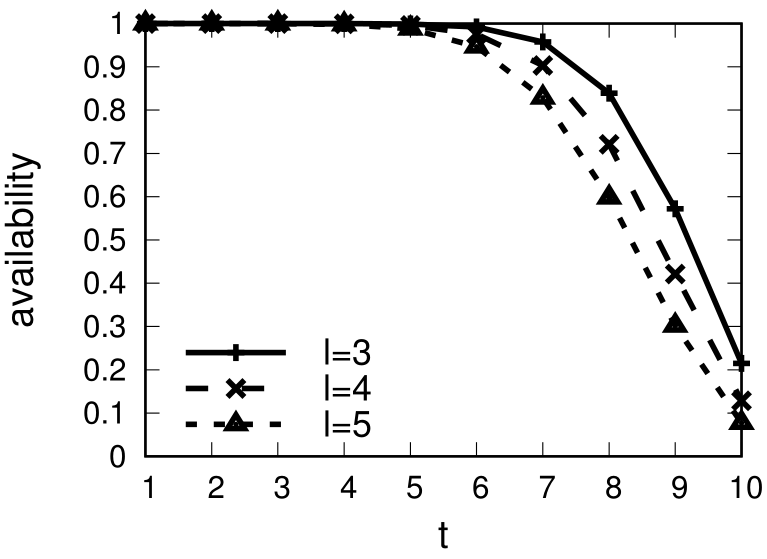}
 }
%\captionsetup{font=bf}
%\vspace{-3mm}
\caption{Service availability when mailman availability is 95\%}
\label {hipc_0102}
%\vspace{-5mm}
\end{figure}

\subsection{Availability}
It is possible that a recruited mailman violates the protocol inadvertently, such as being absent during \textit{TIDS.deliver} due to unexpected network issues.
Considering that the decentralized service providers, namely the mailmen are not as professional as a centralized service provider, a very strict protocol that demands each mailman to maintain 99.9\% availability would be unpractical and may make the bar of being a mailman too high.
On the other hand, we would not want the bar to stay too low because otherwise each service needs to recruit a great number of mailmen to succeed.
A good balance can be made when mailmen can maintain relatively high availability while \textit{TIDS} can achieve very high availability, such as three nine or even four nine availability.

\begin{lemma}
\label{lemma2}
\textit{SilentDelivery} can achieve three nine or even four nine availability when mailmen maintain 95\% availability.
\end{lemma}

\begin{proof}
By denoting mailman availability as $A_{T}$, the service availability $A_{S}$ with parameters $[l,t,n]$ can be computed through the Cumulative Distribution Function of Binomial distribution:
\begin{equation}
A_{S}=1-\sum_{i=n-t+1}^{n} \binom{n}{i} P^i (1-P)^{n-i}
\end{equation}
% $$A_{S}=1-\sum_{i=n-t+1}^{n} \binom{n}{i} P^i (1-P)^{n-i}$$
where $P=1-A_{T}^l$ represents the probability that a $share$ is lost.
As illustrated in Fig.~\ref{hipc_0102}, when $[t,n]=[4,10]$ and $A_{T}=95\%$, we get $A_{S}=99.9\%$ (99.99\%) by taking $l=4$ (3).
 % and 99.99\% with $l=3$.
\end{proof}

We can see that there is a tradeoff between attack resilience and availability offered by \textit{SilentDelivery}.
A larger $l$ results in higher resilience against Sybil attacks while a smaller $l$ improves availability.
We suggest using a smaller $l$ (e.g., $[l,t,n]=[3,4,10]$) when senders are less concerned about attacks regarding premature disclosure because a smaller $l$ can also reduce the number of recruited mailmen and thus $\mathsf{\Xi} remuneration$ that senders need to pay.

\subsection{Cost analysis}
% \textcolor{red}{
We analyze the cost of  \textit{SilentDelivery} as follows:

\begin{lemma}
\label{lemma1}
The \textit{SilentDelivery} protocol has a total cost in the range of $[O(1),O(n)]$ and the $O(1)$ lower bound could be reached as long as the participants are rational.
\end{lemma}

\begin{proof}
% \textcolor{red}{
The total cost would reach the upper bound $c^{\textit{up}}$ when 
(1) the identities of mailmen are uploaded via $\textit{revealIdentity()}$ onto the blockchain and 
(2) the private keys are uploaded via $\textit{revealPrivkey()}$ onto the blockchain, namely,
% $$O(c^{\textit{up}})$$
\begin{equation}
O(c^{\textit{up}}) \to O(c_{\textit{id}} \cdot n+c_{\textit{pk}} \cdot n+c_{\textit{rest}}) \to O(n)
\end{equation}
% $$O(c^{\textit{up}}) \to O(c_{\textit{id}} \cdot n+c_{\textit{pk}} \cdot n+c_{\textit{rest}}) \to O(n)$$
where $c_{\textit{id}}$ and $c_{\textit{pk}}$ denotes cost of uploading per identity and private key and $c_{\textit{rest}}$ represents the total cost of calling other functions.
In contrast, the total cost would reach the lower bound $c^{\textit{low}}$ when none of $\textit{revealIdentity()}$ and \textit{revealPrivkey()} have been invoked, namely
$O(c^{\textit{low}}) \to O(c_{\textit{rest}}) \to O(1)$.
% \noindent 
Both the functions are countermeasures against dishonest participants by fixing problems made by them and confiscating $deposit$ paid by them, so the violators would gain no positive benefit but only a negative payoff.
Considering that rational adversaries choose to violate protocols only when doing so brings them a positive payoff, rational participants would not choose to lose $deposit$, which in turn would push the total cost to reach its lower bound.
\end{proof}

\subsection{Discussion}
\subsubsection{Deliberate disclosure}
% \noindent 
% \textcolor{red}{
% \textbf{Deliberate disclosure}: 
It is possible that an active adversary who does not care about his or her loss but does care about the loss of the sender, will deliberately make premature disclosure, thus making the gas cost O(n) from O(1).
In this case, the security deposit of the adversary will be confiscated and the confiscated money will be used to compensate for the cost of mode switching.
With this strategy, the overall gas cost could not reach O(1), but it could be decreased by a certain amount of value.
Though it is a challenging problem to guarantee O(1) with such active adversaries, it is an interesting and important aspect of future research on this topic.

\subsubsection{Mailmen Collusion}
% \noindent 
% \textcolor{red}{
% \textbf{Collusion}: 
If the mailmen are in collusion, they may be able to decrypt their shares and find out the sent information.
The protocol allows any mailman to report a premature disclosure.
The reporter could receive a monetary reward from the security deposit paid by the suspected mailman as the incentive.
We design the incentive to incentivize mailmen to betray each other in collusion.
For instance, if two mailmen (say A and B) exchange private keys, immediately after receiving private key B, mailman A could report the disclosure of both key A and B to earn security deposit from mailmen B and also prohibit future reports from mailman B.
As a result, it will never be a good idea to disclose any private key in collusion.
However, if a mailman (say A) chooses to decrypt one layer of the onion and disclose only the decrypted onion without revealing the key, it will become difficult for other mailmen to verify the correctness of the decryption and hence, mailman A could not receive any monetary incentive to do this.

% performs misbehaviors inadvertently, such as forgetting providing the service or losing EOA's private key. 
% Such kinds of inadvertent misbehaviors lead to same results of intentionally performing \textit{absent mailman} misbehavior. 
% If we denote the percentage of EOAs performing inadvertent misbehaviors as $p_{IM}$, the success rate of a schedule with parameters $(l,m,n)$ will be computed through the Cumulative Distribution Function of Binomial distribution, namely $SR=1-\sum_{i=n-m+1}^{n} \binom{n}{i} P^i (1-P)^{n-i}$, where $P=1-(1-p_{IM})^l$ represents the probability that one $share$ is lost.
% In Figure~\ref{hipc_0102}, 
% we present the computed schedule success rate when 5\% of mailmen perform misbehaviors inadvertently.
% Specifically, in Figure~\ref{hipc_01}, by fixing $n$ to 5 and changing $m$ from 1 to 5, it shows that a smaller $m$, namely lower threshold for restoring $key$, performs higher resistance against inadvertent misbehaviors. By further changing $l$ from 3 to 5, we can find that a smaller $l$ offers better resistance against inadvertent misbehaviors. Then, in Figure~\ref{hipc_02}, $n$ is increased to 10. The increment of $n$ enhances the resistance against inadvertent misbehaviors when $m$ and $l$ do not change. Thus, larger $l$ and $n$ while smaller $m$ help maintaining high resistance against inadvertent misbehaviors.

\section{Implementation and evaluation}
In this section, we present the implementation and evaluation of the proposed \textit{SilentDelivery} protocol.
We programmed the protocol as smart contracts using \textit{Solidity}, the most commonly used smart contract programming language. 
% We also implement the protocol over the Ethereum official test network \textit{rinkeby}\footnote{\begin{scriptsize} https://www.rinkeby.io, we plan to make the contracts open source shortly. \end{scriptsize}}.
We also anonymously deployed the contracts to the Ethereum official test network \textit{rinkeby}\footnote{\begin{scriptsize} 
$addr(C_{agent})$: 0xa49d94bf3a7eeF256772b68Bf7D799Aa30F6F342 \\
\hspace*{4mm} $addr(C_{sw})$: \hspace*{2.5mm} 0x1640B660147fD684C37d8ccf1caAB732898c8627    \\
\hspace*{4mm} $addr(C_{sup})$: \hspace*{1.6mm}
0x4CBa2Ab68779d861D594A144Aa1a099a27E917e9
\end{scriptsize}
} 
for evaluation purpose.

% \textcolor{red}{
Similar to recent work on blockchain-based platforms~\cite{dziembowski2019perun,das2018yoda}, the key focus of our evaluation is on measuring gas consumption, namely the amount of transaction fees spent in the protocol. This is due to the fact that the execution complexity in Ethereum is measured via gas consumption.
It is worth noting that, in Ethereum, the amount of gas that a single transaction may spend is bounded by a system parameter and hence, the time overhead of executing functions inside smart contracts is small, usually in the scale of hundreds of milliseconds.
In addition, we compare the gas consumption of \textit{SilentDelivery} with that of \textit{Kimono}~\cite{Kimono}, a recent project that also employed Ethereum to release private information to future via recruiting mailmen.
% \textcolor{red}{
The main protocol in \textit{Kimono} consists of four stages, namely secret splitting, secret fragment distribution, secret fragment publishing and secret reconstruction.
During the first two stages, a sender encrypts private information with a $key$, splits the $key$ to $shares$ and distributes the $shares$ to a group of mailmen.
Then, during the last two stages, the mailmen reveal $shares$ at prescribed time point and restore the $key$. 
Similar to the strawman protocol, in \textit{Kimono}, a sender needs to upload the identities of all the recruited mailmen to the blockchain and later the recruited mailmen need to upload all the $shares$ to the blockchain.
Hence, as we discussed in section 3.3, the premature revelation of recruitment relationships assists collusion among mailmen. 
Also, both the uploaded identities of recruited mailmen and $shares$ make the overall cost $O(n)$.

% Since the primary goal of this paper is to compare the fees charged by our solution with that of existing commercial crowdsourcing intermediate companies (e.g., 99designs), we have designed to mimic the evaluation strategy adopted in these recent blockchain-based platforms[4,5]. 
% }

\begin{table}
\centering
\begin{tabular}{|c |c |p{2.6cm}|p{0.9cm}|p{0.8cm}|} \toprule 
    {\textbf{Phase}} & {\textbf{Step}} & {\textbf{Function}} & {\textbf{Gas}} & {\textbf{USD}} \\ \midrule
    \multirow{2}{*}{\textbf{TIDS.send}}
    & 1 & deploy $C_{sw}$ & 616666 & \$1.81 \\
    & 2 & newService & 83121 & \$0.24 \\ \midrule
    \multirow{2}{*}{\textbf{Epoch-0}}
    & 1.1 & deploySupplementary & 2425356 & \$7.10 \\
    & 1.2 & reportPremature & 65317 & \$0.19 \\ \midrule
    \textbf{Epoch-1} & 4.2 & recipientReceipt & 54291 & \$0.16 \\ \midrule
    \multirow{2}{*}{\textbf{Epoch-2}}
    & 5 & deploySupplementary & 2425356 & \$7.10 \\
    & 7.2 & revealIdentity & 72678 & \$0.21 \\ \midrule
    \textbf{Epoch-3} & 8 & revealPrivkey & 90689 & \$0.27 \\ \midrule
    \multirow{3}{*}{\textbf{Epoch-4}}
    & 9 & reportAbsent & 65343 & \$0.19 \\
    & 10 & reportFake & 1280723 & \$3.75 \\
    & 11 & informAgent & 57042 & \$0.17 \\ \midrule
    \textbf{Epoch-5} & 12 & recipientReceipt & 54291 & \$0.16 \\ 
    \bottomrule
\end{tabular}
%\vspace{-1mm}
\caption{
Key functions and their cost in Gas and USD. 
}
\vspace{-4mm}
\label{t2}
\end{table} 

\subsection{Gas consumption of \textit{SilentDelivery}}

% We first present the cost of key functions in the programmed smart contracts.
In TABLE~\ref{t2}, we list the key functions in the programmed smart contracts that interact with protocol participants during different phases of \textit{SilentDelivery} and the cost of these functions in both Gas and USD.
The cost in USD was computed through 
$cost(USD)=cost(Gas)*GasToEther*EtherToUSD$,
where $GasToEther$ and $EtherToUSD$ were taken as their mean value during the first half of the year 2019 recorded in \textit{Etherscan}
\footnote{\begin{scriptsize} https://etherscan.io/charts \end{scriptsize}}, 
which are $1.67*10^{-8}$ Ether/Gas and 175 USD/Ether, respectively.
As illustrated by the results, in the lightweight mode, the completion of a service only requires a sender to deploy $C_{sw}$ (\$1.81) and set up a new service (\$0.24) during \textit{TIDS.send} and a recipient to submit the receipt (\$0.16) during \textit{Epoch-1}, which costs only \$2.21 in total.
This cost is independent of the number of recruited mailmen, namely $O(1)$, so the tradeoff between security and scalability is eliminated and the scalability of \textit{TIDS} gets significantly improved.
In contrast, the heavyweight mode requires a sender to deploy $C_{sw}$ (\$1.81) and set up service (\$0.24) during \textit{TIDS.send} and a mailman to deploy $C_{sup}$ (\$7.10) and reveal identities of all recruited mailmen during \textit{Epoch-2} ($\$0.21 n$). It then requires all recruited mailmen to reveal $privkeys$ during \textit{Epoch-3} ($\$0.27 n$) and finally a recipient to submit the receipt (\$0.16) during \textit{Epoch-5}. It thus costs $\$(9.31+0.48n)$ for completing a service that recruits $n$ mailmen.
If any misbehavior occurs, the reporting functions can be invoked during \textit{Epoch-1} and \textit{Epoch-4} and the cost for calling these reporting functions will be deducted from $\mathsf{\Xi} deposit$ paid by protocol violators.

% \textcolor{red}{
The overall cost of using \textit{SilentDelivery} in its lightweight mode is \$2.21 and the fee is charged for employing the solid decentralized security in Ethereum. It is worth noting that the fee was computed using a half-year average price of Ether and sometimes the price of Ether was insanely high. In practice, users could perchase Ether with much lower exchange rate, and the lower price could help reduce the cost of using \textit{SilentDelivery}. 
We next analyze the relationship between the cost of using \textit{SilentDelivery} and the scale of recruited mailmen. In lemma 2, we proved that the minimum amount of security deposit that an adversary needs to deposit to obtain $t$ $shares$ on average via controlled mailmen would be $(l-1)vd$, which is independent of the number of recruited mailmen. In lemma 3, we get service availability 99.9\% (99.99\%) by taking the factore l with 4(3). 
Also, we have seen that the overall fee in its lightweight mode is not increased when the number of recruited mailmen get scaled up.
Therefore, an amount of 30 recruited mailmen (e.g., $[l,t,n]=[3,4,10]$) could give consideration to both security and cost.
In future work, we will further decouple the security and cost of textit{SilentDelivery} from the scale of the participants of the protocol.

\begin{figure}
\centering
{
   
    \includegraphics[width=0.75\columnwidth]{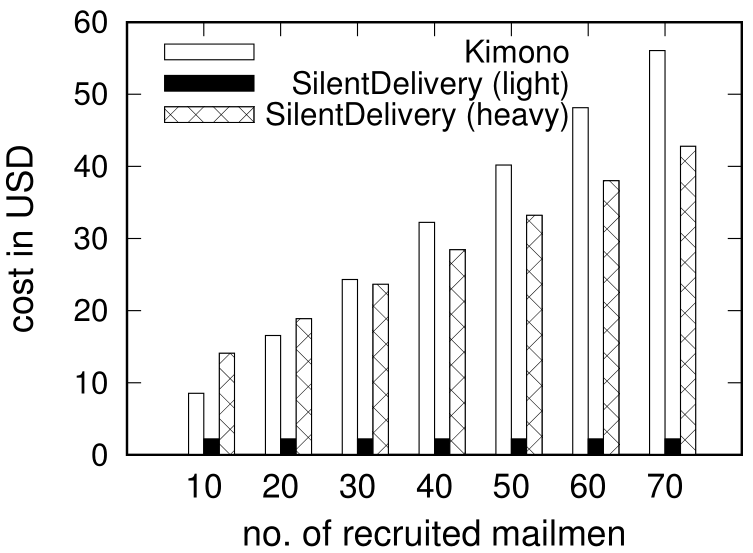}
}
\vspace{-2mm}
\caption {Cost of \textit{Kimono} and \textit{SilentDelivery}}
\vspace{-3mm}
\label{cost} 
\end{figure}

\subsection{Comparison between \textit{SilentDelivery} and \textit{Kimono}}

Next, we compare the gas consumption of \textit{SilentDelivery} with that of \textit{Kimono}~\cite{Kimono}.
For the purpose of evaluating the scalability of the protocols, we tested the gas consumption of recruiting different numbers of mailmen in the two projects and displayed the results in Fig.~\ref{cost}.
As can be seen, compared with \textit{Kimono}, the \textit{SilentDelivery} protocol proposed in this paper is much more scalable and cost-effective.
Due to the removal of non-scalable regulations, the lightweight mode offers $O(1)$ cost, which for the first time makes a decentralized approach for timed-delivery of private information practical and scalable.
We noticed that even the heavyweight mode is more scalable than \textit{Kimono}.
This is due to the use of \textit{silent recruitment}, which reduces the cost of the protocol during \textit{TIDS.send}.
In summary, compared with the state of the art, both the modes offered by \textit{SilentDelivery} are more scalable.
Particularly, the lightweight mode of \textit{SilentDelivery} can reduce the cost of recruiting more than 20 mailmen by over 85\%.

% Furthermore, we believe that the much lower $O(1)$ cost offered by the lightweight mode would incentivize mailmen to stay honest, which can reduce the cost of \textit{Kimono} by over 85\% when more than 20 mailmen are recruited.

\section{Related work}

\subsection{Sending private information to a future time point}
% The study of timed release of private information began with May~\cite{may1992timed} and has subsquently led to the development of four distinct types of solutions namely time-lock puzzles~\cite{rivest1996time,mahmoody2011time,bitansky2016time}, time servers~\cite{blake2004scalable,emura2011timed,kasamatsu2012time}, blockchain puzzles~\cite{liu2018build,liu2015time,liu2015extractable} and smart contract-based release~\cite{ning2018keeping,li2018decentralized,Kimono}.
% Our work in this paper tackles the key limitations of the state-of-the-art approaches using blockchain-based smart contracts~\cite{ning2018keeping,li2018decentralized,Kimono}. To the best of our knowledge,  \textit{SilentDelivery} is the \textit{first} decentralized solution for the timed data release problem with strong guarantees on security, scalability and cost-efficiency. 

The study of timed release of private information began with May~\cite{may1992timed}.
Since then, there have been extensive studies on this problem.
One representative approach~\cite{rivest1996time,mahmoody2011time,bitansky2016time} protects private information with a time-lock puzzle, forcing recipients to solve a cryptographic puzzle to obtain the information. 
Nevertheless, the time for solving such puzzles is non-deterministic and as a result, the delivery time of the information can not be precisely controlled. Also, cryptographic approaches for timed data release come with a very significant computational cost and as such, these techniques are not scalable.
Another well-studied approach~\cite{blake2004scalable,emura2011timed,kasamatsu2012time} relies on a (semi-)trusted time server to release time trapdoors to recipients at specified future time points.  These techniques involve a single point of trust and create a safety bottleneck.
The recent emergence of blockchain technology~\cite{nakamoto2008bitcoin} and smart contracts~\cite{wood2014ethereum} have started the development of new decentralized approaches.
One of the decentralized approaches~\cite{liu2018build,liu2015time,liu2015extractable} encloses private information with blockchain puzzles used in Proof-of-Work~\cite{nakamoto2008bitcoin} and therefore  minimizes the computational burden of the information senders as the blockchain puzzles are periodically solved by blockchain miners.
However in such an approach, the involved heavy cryptographic primitives result in very high performance overhead~\cite{liu2015time,ning2018keeping}. 
Another direction of recent decentralized techniques for timed data release~\cite{ning2018keeping,li2018decentralized,Kimono} leverages smart contracts~\cite{wood2014ethereum} to establish a decentralized virtual autonomous agent, through which an information sender could recruit a group of peers from the blockchain network as her mailmen to cooperatively maintain and deliver her private information to recipients.
% \textcolor{red}{ 
Here, the transparency of smart contracts makes it difficult to conceal any information recorded by the virtual agent and therefore challenges the service security in multiple aspects. Besides that, the cost of running smart contracts in this approach is proportional with the number of recruited mailmen, making it costly to recruit more mailmen to gain higher service availability. 
% }
Our work in this paper tackles the key limitations of the state-of-the-art approaches using blockchain-based smart contracts~\cite{ning2018keeping,li2018decentralized,Kimono}. To the best of our knowledge,  \textit{SilentDelivery} is the \textit{first} decentralized solution for the timed data release problem with strong guarantees on security, scalability and cost-efficiency.

\subsection{Scaling blockchain with off-chain execution}
Off-chain execution of smart contracts is a promising solution for improving blockchain scalability~\cite{cheng2019ekiden,das2019fastkitten,dziembowski2018general}.
However, recent works in this line have to either assume one honest manager for off-chain execution~\cite{cheng2019ekiden} or allow the execution to get aborted when the manager is dishonest~\cite{das2019fastkitten}.
The state channel network (SCN)~\cite{dziembowski2018general} could achieve the never abort property when at least one participant is honest but it only supports two-participant contracts.
The \texttt{NF-Crowd} protocols proposed in this paper extend the objective to support complex multi-participant multi-round smart contracts without losing the never abort property.
In addition, \textit{SilentDelivery} has been implemented in Ethereum, so it is ready-to-use.

% Also, inspired by the recently advanced techniques for improving blockchain scalability, including payment channel network (PCN)~\cite{dziembowski2017perun} and state channel network (SCN)~\cite{dziembowski2018general}

\subsection{Using cryptocurrency as security deposits}
There have been many recent efforts on blockchain-based protocol design that leverage cryptocurrency as security deposits to penalize unexpected behaviors and improve security. Andrychowicz \textit{et al.}~\cite{andrychowicz2014secure} used bitcoin to penalize anyone who unfairly aborts a secure multiparty computation (SMC).
Kiayias \textit{et al.}~\cite{kiayias2015traitor} proposed to use bitcoin as collateral to protect digital content.
In~\cite{dong2017betrayal}, the authors use Ether as security deposits to provide verifiable cloud computing.
Matsumoto \textit{et al.}~\cite{matsumoto2017ikp} used Ether as security deposits to enforce certificate authorities to be honest.
Inspired by these previous efforts, \textit{SilentDelivery} demands each mailman lock Ether in smart contracts as security deposits to penalize potential misbehaviors violating the protocol and thereby enforces recruited mailmen to stay honest.

% {\bf Except subsection A, all other subsections are discussing related work for the techniques used in the paper. They are not discussion related work for the problem solved in this paper. I think some text from the intro should be moved to related work.}

% \textcolor{purple}{\textbf{
% I have merged paragraph 2 of section 1 with subsection A of section 2. The red parts may need further editing. 
% Please let me know if you have other comments about this.
% }}

\section{Conclusion}
This paper proposes \textit{SilentDelivery}, a practical decentralized solution for cost-effectively implementing timed-delivery of private information. Our solution employs a novel combination of threshold secret sharing and decentralized smart contracts and tackles two key challenges that significantly limit the security and scalability of the protocol.
Through \textit{silent recruitment}, \textit{SilentDelivery} makes a mailman get recruited by a sender \textit{silently} without the knowledge of any third party while still making it possible for the recruitment relationship to be revealed to the smart contracts during a future time-frame.
Through \textit{dual-mode execution}, \textit{SilentDelivery} incentivizes mailmen to make the protocol get executed in the lightweight mode, reducing the service cost significantly.
We rigorously analyze the security of \textit{SilentDelivery} and implement the protocol over the Ethereum official test network. 
The results demonstrate that \textit{SilentDelivery} reduces the cost of running smart contracts by over 85\% and is more secure and scalable compared to the state of the art.

\section*{Acknowledgement}
Balaji Palanisamy is supported by a grant (Award \#2020071) from the US National Science Foundation (NSF) SaTC program. 
Chao Li acknowledges the partial support through the Fundamental Research Funds for the Central Universities (Grant No. 2019RC038).

\renewcommand\refname{Reference}
\bibliographystyle{IEEEtran}
\bibliographystyle{plain}
\urlstyle{same}

\bibliography{main.bib}

\vskip 0pt plus -1fil

\begin{IEEEbiography} [{\includegraphics[width=1in,height=1.25in]{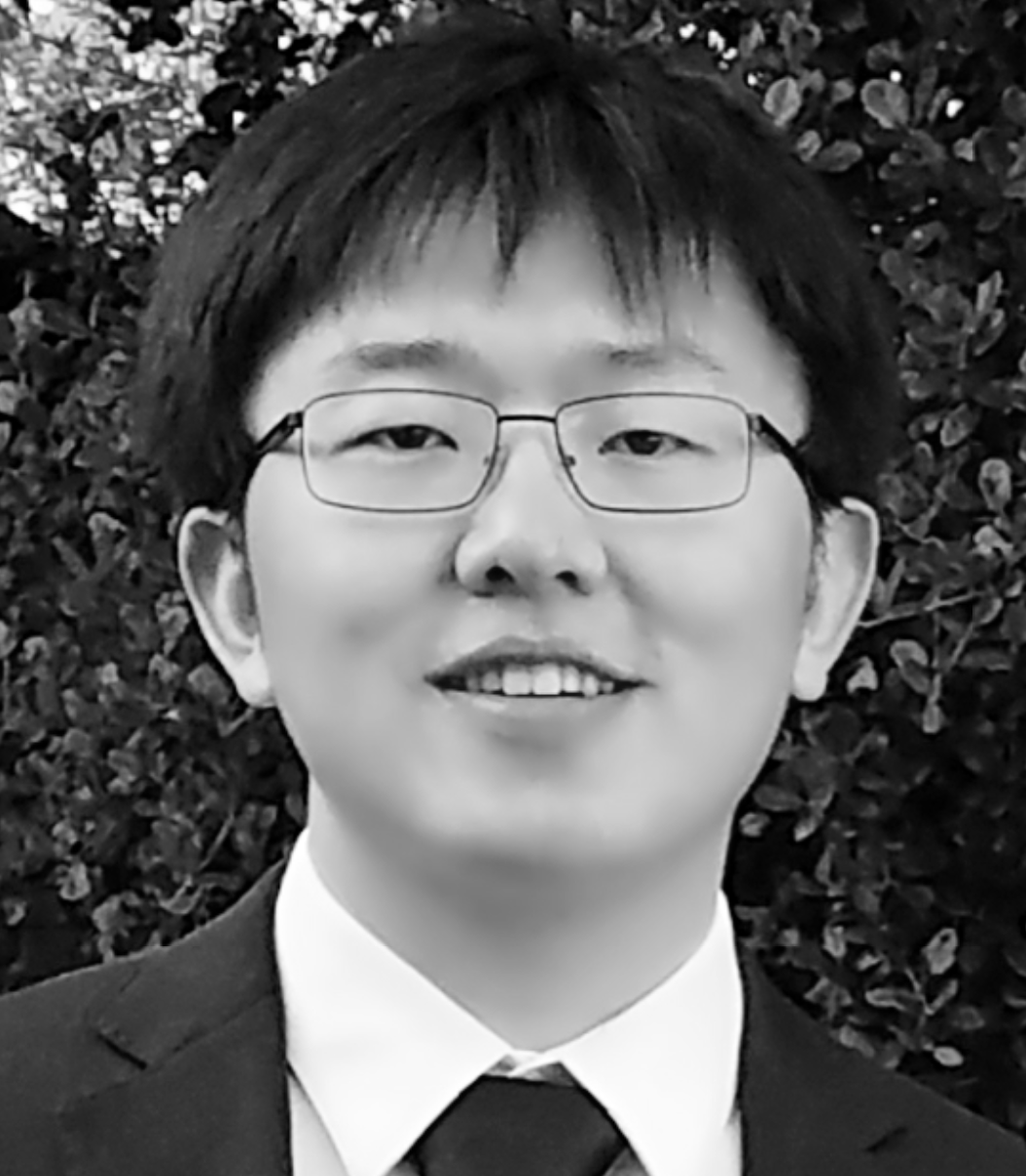}}]{Chao Li}
    is an Assistant Professor in the School of Computing and Information Technology at Beijing Jiaotong University. He received his Ph.D. degree from the School of Computing and Information at University of Pittsburgh and his MSc degree from Imperial College London. His current research interests are focused on Blockchain and Data Privacy. 
    % He is a member of the IEEE.
\end{IEEEbiography}

\vskip 0pt plus -1fil

\begin{IEEEbiography}
    [{\includegraphics[width=1in,height=1.25in]{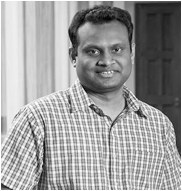}}]{Balaji Palanisamy}
is an Associate Professor in the School of computing and information in University of Pittsburgh. He received his M.S and Ph.D. degrees in Computer Science from the college of Computing at Georgia Tech in 2009 and 2013, respectively. His primary research interests lie in scalable and privacy-conscious resource management for large-scale Distributed and Mobile Systems. At University of Pittsburgh, he codirects research in the Laboratory of Research and Education on Security Assured Information Systems (LERSAIS).
 % He is a member of the IEEE. Dr. Palanisamy is currently serving as an Associate Editor for the IEEE Transactions on Services Computing journal.
\end{IEEEbiography}

\end{document}